\def\fullversion{1}   % Important: for camera-ready, set to 0; for full-length with appendix, set to 1

\documentclass{sig-alternate-2013}
\usepackage{amsmath,amsfonts,amssymb,microtype}
\usepackage{graphicx}
\usepackage{algorithm,algpseudocode}
\usepackage{caption}
\usepackage{array}
\usepackage{thmtools}
\usepackage{thm-restate}
\usepackage{enumitem}
\usepackage{titletoc}

\newcommand{\fixnewfont}[6]{%
  \renewcommand#1{\fontsize{#5}{#6}\usefont{\encodingdefault}{#2}{#3}{#4}}%
}

\fixnewfont{\secfnt}{ptm}{b}{n}{12}{14}
\fixnewfont{\secit}{ptm}{b}{it}{12}{14}
\fixnewfont{\subsecfnt}{ptm}{m}{it}{11}{14}
\fixnewfont{\subsecit}{ptm}{b}{it}{11}{14}
\fixnewfont{\ttlfnt}{phv}{b}{n}{18}{20}
\fixnewfont{\ttlit}{phv}{b}{sl}{18}{20}
\fixnewfont{\subttlfnt}{phv}{m}{n}{14}{20}
\fixnewfont{\subttlit}{phv}{m}{sl}{14}{20}
\fixnewfont{\subttlbf}{phv}{b}{n}{14}{20}
\fixnewfont{\aufnt}{phv}{m}{n}{12}{16}
\fixnewfont{\auit}{phv}{m}{sl}{12}{16}
\fixnewfont{\affaddr}{phv}{m}{n}{10}{12}
\fixnewfont{\affaddrit}{phv}{m}{sl}{10}{12}
\fixnewfont{\eaddfnt}{phv}{m}{n}{12}{14}
\fixnewfont{\ixpt}{ptm}{m}{n}{9}{11}
\fixnewfont{\confname}{ptm}{m}{it}{8}{10}
\fixnewfont{\crnotice}{ptm}{m}{n}{8}{10}
\fixnewfont{\ninept}{ptm}{m}{n}{9}{10.5}
\newcommand{\Omit}[1]{}

\declaretheorem[name=Theorem,numberwithin=section]{theorem}
\newtheorem{itheorem}{Summary Theorem}
\declaretheorem[name=Lemma,numberwithin=section]{lemma}

\newtheorem{corollary}{Corollary}
\declaretheorem[name=Observation,numberwithin=section]{observation}

\newcommand{\bigO}[1]{O\left(#1\right)}
\newcommand{\bigOmega}[1]{\Omega\left(#1\right)}
\newcommand{\bigTheta}[1]{\Theta\left(#1\right)}

\DeclareMathOperator{\E}{\mathbb{E}}   % expectation

\DeclareMathAlphabet{\mathbbold}{U}{bbold}{m}{n}
\newcommand{\1}{\mathbbold{1}}           % indicator

\newfont{\mycrnotice}{ptmr8t at 7pt}
\newfont{\myconfname}{ptmri8t at 7pt}
\let\crnotice\mycrnotice%
\let\confname\myconfname%
\permission{Permission to make digital or hard copies of all or part of this work for personal or classroom use is granted without fee provided that copies are not made or distributed for profit or commercial advantage and that copies bear this notice and the full citation on the first page. Copyrights for components of this work owned by others than the author(s) must be honored. Abstracting with credit is permitted. To copy otherwise, or republish, to post on servers or to redistribute to lists, requires prior specific permission and/or a fee. Request permissions from permissions@acm.org.}
\conferenceinfo{ITCS'15,}{January 11--13, 2015, Rehovot, Israel. \\
{\mycrnotice{Copyright is held by the owner/author(s). Publication rights licensed to ACM.}}}
\copyrightetc{ACM \the\acmcopyr}
\crdata{978-1-4503-3333-7/15/01\ ...\$15.00.\\
http://dx.doi.org/10.1145/2688073.2688095
}

\begin{document}
\title{$\ell_p$ Testing and Learning of Discrete Distributions}
\numberofauthors{1}
\author{
\alignauthor Bo Waggoner \\
       \affaddr{Harvard} \\
       \email{bwaggoner@fas.harvard.edu}
} % author
\maketitle

\begin{abstract}
The classic problems of testing uniformity of and learning a discrete distribution, given access to independent samples from it, are examined under general $\ell_p$ metrics. The intuitions and results often contrast with the classic $\ell_1$ case. For $p > 1$, we can learn and test with a number of samples that is independent of the support size of the distribution: With an $\ell_p$ tolerance $\epsilon$, $O(\max\{ \sqrt{1/\epsilon^q}, 1/\epsilon^2 \})$ samples suffice for testing uniformity and $O(\max\{ 1/\epsilon^q, 1/\epsilon^2\})$ samples suffice for learning, where $q=p/(p-1)$ is the conjugate of $p$. As this parallels the intuition that $O(\sqrt{n})$ and $O(n)$ samples suffice for the $\ell_1$ case, it seems that $1/\epsilon^q$ acts as an upper bound on the ``apparent'' support size.

For some $\ell_p$ metrics, uniformity testing becomes easier over larger supports: a 6-sided die requires fewer trials to test for fairness than a 2-sided coin, and a card-shuffler requires fewer trials than the die. In fact, this inverse dependence on support size holds if and only if $p > \frac{4}{3}$. The uniformity testing algorithm simply thresholds the number of ``collisions'' or ``coincidences'' and has an optimal sample complexity up to constant factors for all $1 \leq p \leq 2$. Another algorithm gives order-optimal sample complexity for $\ell_{\infty}$ uniformity testing. Meanwhile, the most natural learning algorithm is shown to have order-optimal sample complexity for all $\ell_p$ metrics.

The author thanks Cl\'{e}ment Canonne for discussions and contributions to this work.

\ifnum\fullversion=1
This is the full version of the paper appearing at ITCS 2015.
\fi
\end{abstract}




%%%%%%%%%%%%%%%%%%%%%%%%%%%%%%%%%%%%%%%%%%%%%%
\section{Introduction}
Given independent samples from a distribution, what we can say about it? This question underlies a broad line of work in statistics and computer science. Specifically, we would like algorithms that, given a small number of samples, can test whether some property of the distribution holds or can learn some attribute of the distribution.

This paper considers two natural and classic examples. \emph{Uniformity testing} asks us to decide, based on the samples we have drawn, whether the distribution is uniform over a domain of size $n$, or whether it is ``$\epsilon$-far'' from uniform according to some metric. \emph{Distribution learning} asks that, given our samples, we output a sketch or estimate that is within a distance $\epsilon$ of the true distribution. For both problems, we would like to be correct except with some constant probability of failure (\emph{e.g.} $\frac{1}{3}$). The question studied is the number of independent samples required to solve these problems.

In practical applications we might imagine, such as a web company wishing to quickly test or estimate the distribution of search keywords in a given day, the motivating goal is to formally guarantee good results while requiring as few samples as possible. Under the standard $\ell_1$ distance metric (which is essentially equivalent to total variation distance -- we will use the term $\ell_1$ only in this paper), the question of uniformity testing over large domains was considered by Paninski~\cite{paninski2008coincidence}, showing that $\bigTheta{\frac{\sqrt{n}}{\epsilon^2}}$ samples are necessary and sufficient for testing uniformity on support size $n$, and it is known by ``folklore'' that $\bigTheta{\frac{n}{\epsilon^2}}$ samples are necessary and sufficient for learning. Thus, these questions are settled\footnote{\cite{paninski2008coincidence} focused on the regime where support size is very large, so order-optimal $\ell_1$ uniformity testing for the case of smaller $n$ may have been technically open prior to this work.} (up to constant factors) if we are only interested in $\ell_1$ distance.

However, in testing and learning applications, we may be interested in other choices of metric than $\ell_1$. And more theoretically, we might wonder whether the known $\ell_1$ bounds capture all of the important intuitions about the uniformity testing and distribution learning problems. Finally, we might like to understand our approaches for the $\ell_1$ metric in a broader context or seek new techniques. This paper addresses these goals via $\ell_p$ metrics.

\subsection{Motivations for $\ell_p$ Metrics}
In the survey ``Taming Big Probability Distributions''~\cite{rubinfeld2012taming}, Rubinfeld notes that even sublinear bounds such as the above $\bigTheta{\frac{\sqrt{n}}{\epsilon^2}}$ may still depend unacceptably on $n$, the support size. If we do not have enough samples, Rubinfeld suggests possible avenues such as assuming that the distribution in question has some very nice property, \emph{e.g.} monotonicity, or assuming that the algorithm has the power to make other types of queries.

However, it is still possible to ask what can be done without such assumptions. One answer is to consider what we can say about our data under other measures of distance than the $\ell_1$ distance. Do fewer samples suffice to draw conclusions? A primary implication of this paper's results is that this approach does succeed under general $\ell_p$ metrics. The $\ell_p$ distance between two probability distributions $A,B \in \mathbb{R}^n$ for any $p \geq 1$, where $A_i$ is the probability of drawing coordinate $i$ from distribution $A$, is the $\ell_p$ norm of the vector of differences in probabilities:
 \[ \|A-B\|_p = \left(\sum_{i=1}^n |A_i - B_i|^p \right)^{1/p} . \]
The $\ell_{\infty}$ distance is the largest difference of any coordinate, \emph{i.e.} $\|A-B\|_{\infty} = \max_i |A_i - B_i|$.

Unlike the $\ell_1$ case, it will turn out that for $p > 1$, we can draw conclusions about our data with a number of samples that is \emph{independent of $n$} and depends only on the desired error tolerance $\epsilon$. We also find smaller dependences on the support size $n$; in fact, for uniformity testing we find sometimes (perhaps counterintuitively) that there is an \emph{inverse} dependence on $n$. The upshot is that, if we have few samples, we may not be able to confidently solve an $\ell_1$ testing or learning problem, but we may have enough data to draw conclusions about, say, $\ell_{1.5}$ distance. This may also be useful in saying something about the $\ell_1$ case: If the true distribution $A$ has small $\ell_{1.5}$ distance from our estimate $\hat{A}$, yet actually does have large $\ell_1$ distance from $\hat{A}$, then it must have a certain shape (\emph{e.g.} large support with many ``light hitters'').\footnote{I thank the anonymous reviewers for suggestions and comments regarding this motivation, including the $\ell_{1.5}$ example.}

Thus, this is the first and primary motivation for the study of $\ell_p$ metrics: to be able to draw conclusions with few samples but without making assumptions.

A second motivation is to understand learning and testing under other $\ell_p$ metrics for their own sake. In particular, the $\ell_2$ and $\ell_{\infty}$ cases might be considered important or fundamental. However, even these are not always well understood. For instance, ``common knowledge'' says that $\bigTheta{\frac{1}{\epsilon^2}}$ samples are required to determine if one side of a coin is $\epsilon$-more likely to come up than it should be; one might naively think that the same number of trials are required to test if any card is $\epsilon$-more likely to be top in a shuffle of a sorted deck. But the latter can be far less, as small as $\bigTheta{\frac{1}{\epsilon}}$ (depending on the relationship of $\epsilon$ to the support size), so a large improvement is possible.

Other $\ell_p$ norms can also be of interest when different features of the distribution are of interest. These norms trade off between measuring the \emph{tail} of the distribution ($p=1$ measures the total deviation even if it consists of many tiny pieces) and measuring the \emph{heavy} portion of the distribution ($p=\infty$ measures only the single largest difference and ignores the others). Thus, an application that needs to strike a balance may find that it is best to test or estimate the distribution under the particular $p$ that optimizes some tradeoff.

General $\ell_p$ norms, and especially $\ell_2$ and $\ell_{\infty}$, also can have immediate applications toward testing and learning other properties. For instance, \cite{batu2013testing} developed and used an $\ell_2$ tester as a black box in order to test the $\ell_1$ distance between two distributions. Utilizing a better $\ell_2$ tester (for instance, one immediately derived from the learner in this paper) leads to an immediate improvement in the samples required by their algorithm for the $\ell_1$ problem.\footnote{Further improvement for this problem is achieved in \cite{chan2014optimal}.}

A third motivation for $\ell_p$ testing and learning, beyond drawing conclusions from less data and independent interest/use, is to develop a deeper understanding of $\ell_p$ spaces and norms in relation to testing and learning problems. Perhaps techniques or ideas developed for addressing these problems can lead to more simple, general, and/or sharp approaches in the special $\ell_1$ case. More broadly, learning or sketching general $\ell_p$ vectors have many important applications in settings such as machine learning (\emph{e.g.} \cite{kloft2011lp}), are of independent interest in settings such as streaming and sketching (\emph{e.g.} \cite{indyk2006stable}), and are a useful tool for estimating other quantities (\emph{e.g.} \cite{cormode2003comparing}). Improved understandings of $\ell_p$ questions have been used in the past to shed new light on well-studied $\ell_1$ problems \cite{lee2004embedding}. Thus, studying $\ell_p$ norms in the context of learning and testing distributions may provide the opportunity to apply, refine, or develop techniques relevant to these areas.

\subsection{Organization}
The next section summarizes the results and describes some of the key intuitions/conceptual takeaways from this work. Then, we will describe the results and techniques for the uniformity testing problem, and then the learning problem. We then conclude by discussing the broader context, prior work, and future work.

Most proofs are omitted in the body of the paper (though sketches are usually provided).
\ifnum\fullversion=1
There is attached an appendix containing all proofs.
\else
The full version of the paper contains an appendix with all proofs.
\fi

\section{Summary and Key Themes} \label{section:summary}
At a technical level, this paper proves upper and lower bounds for number of samples required for testing uniformity and learning for $\ell_p$ metrics. These problems are formally defined as follows. For each problem, we are given i.i.d. samples from a distribution $A$ on support size $n$. The algorithm must specify the number of samples $m$ to draw and satisfy the stated guarantees.

\textbf{Uniformity testing:} If $A = U_n$, the uniform distribution on support size $n$, then output ``uniform''. If $\|A-U_n\|_p \geq \epsilon$, then output ``not uniform''. In each case, the output must be correct except with some constant failure probability $\delta$ (\emph{e.g.} $\delta = \frac{1}{3}$).

\textbf{Learning:} Output a distribution $\hat{A}$ satisfying that $\|A-\hat{A}\|_p \leq \epsilon$. This condition must be satisfied except with some constant failure probability $\delta$ (\emph{e.g.} $\delta = \frac{1}{3}$).

In both cases, the algorithm is given $p,n,\epsilon,\delta$.

\begin{itheorem} For the problems of testing uniformity of and learning a distribution, the number of samples necessary and sufficient satisfy, up to constant factors depending on $p$ and $\delta$, the bounds in Table \ref{tab:results-summary}.

In particular, for each fixed $\ell_p$ metric and failure probability $\delta$, the upper and lower bounds match up to a constant factor for distribution learning for all parameters and for uniformity testing when $1 \leq p \leq 2$, when $p = \infty$, and when $p > 2$ and $n$ is ``large'' ($n \geq \frac{1}{\epsilon^2}$).
\end{itheorem}
Table \ref{tab:results-summary} is intended as a reference and summary; the reader can safely skip it and read on for a description and explanation of the key themes and results, after which (it is hoped) Table \ref{tab:results-summary} will be more comprehensible.

Later in the paper, we give more specific theorems containing (small) explicit constant factors for our algorithms.

Some of these bounds are new and employ new techniques, while others are either already known or can be deduced quickly from known bounds; discussion focuses on the novel aspects of these results and Section \ref{section:related-work} describes the relationship to prior work.

The remainder of this section is devoted to highlighting the most important themes and conceptually important or surprising results (in the author's opinion). The following sections detail the techniques and results for the uniformity testing and learning problems respectively.

\begin{table}
{ % tabular spacing
\renewcommand{\arraystretch}{2.2}
\renewcommand{\tabcolsep}{10pt}
%\fbox{ \parbox{\linewidth}{
\hrule
\vspace{1pt}
\hrule
\vspace{10pt}
\textbf{Learning for $1 \leq p \leq 2$:}
\vspace{1pt}

\hspace{10pt}   \begin{tabular}{p{52pt}|cc}
   regime & $n \leq \frac{1}{\epsilon^q}$             & $n \geq \frac{1}{\epsilon^q}$  \\
   \hline
   necessary and sufficient  & {\Large $\frac{n}{(n^{1/q}\epsilon)^2}$}  & {\Large $\frac{1}{\epsilon^q}$}  \\
   \end{tabular}

\vspace{20pt}
\textbf{Uniformity testing for $1 \leq p \leq 2$:}
\vspace{1pt}

\hspace{10pt}   \begin{tabular}{p{52pt}|cc}
   regime & $n \leq \frac{1}{\epsilon^q}$                    & $n \geq \frac{1}{\epsilon^q}$  \\
   \hline
   necessary and sufficient  & {\Large $\frac{\sqrt{n}}{(n^{1/q}\epsilon)^2}$}  & {\Large $\sqrt{\frac{1}{\epsilon^q}}$}
   \end{tabular}

\vspace{20pt}
\textbf{Learning for $2 \leq p \leq \infty$:}\\
$~~~~~~~$ $\frac{1}{\epsilon^2}$ (necessary and sufficient, all regimes).

\vspace{20pt}
\textbf{Uniformity testing for $p = \infty$:}
\vspace{1pt}

\hspace{10pt}   \begin{tabular}{p{52pt}|ccc}
   regime & $\bigTheta{\frac{n}{\ln(n)}} \leq \frac{1}{\epsilon}$  & $\frac{1}{\epsilon} \leq \bigTheta{\frac{n}{\ln(n)}}$  \\
   \hline
   necessary and sufficient  & {\Large $\frac{\ln(n)}{n \epsilon^2}$} & {\Large $\frac{1}{\epsilon}$}  \\
   \end{tabular}

\vspace{20pt}
\textbf{Uniformity testing for $2 < p < \infty$:}
\vspace{1pt}

\renewcommand{\tabcolsep}{7pt}
\hspace{10pt}   \begin{tabular}{p{32pt}|cp{60pt}c}
   regime & $\bigTheta{\frac{n}{\ln(n)}} \leq \frac{1}{\epsilon}$  & $\frac{1}{\epsilon} \leq \bigTheta{\frac{n}{\ln(n)}}$, $n \leq \frac{1}{\epsilon^2}$  & $n \geq \frac{1}{\epsilon^2}$  \\
   \hline
   necessary   & {\Large $\frac{\ln(n)}{n \epsilon^2}$}  & {\Large $\frac{1}{\epsilon}$}  & {\Large $\frac{1}{\epsilon}$}  \\
   sufficient  & {\Large $\frac{1}{\sqrt{n}\epsilon^2}$} & {\Large $\frac{1}{\sqrt{n}\epsilon^2}$} & {\Large $\frac{1}{\epsilon}$}  
   \end{tabular}
%\end{itemize}
\vspace{5pt}
\hrule
\vspace{1pt}
\hrule
\vspace{8pt}
} % tabular spacing
\caption[]{\textbf{Results summary.} In each problem, we are given independent samples from a distribution on support size $n$. Each entry in the tables is the number of samples drawn necessary and/or sufficient, up to constant factors depending only on $p$ and the fixed probability of failure. Throughout the paper, $q$ is the H\"{o}lder conjugagte of $p$, with $q = \frac{p}{p-1}$ (and $q = \infty$ for $p=1$).

In \emph{uniformity testing}, we must decide whether the distribution is $U_n$, the uniform distribution on support size $n$, or is $\ell_p$ distance at least $\epsilon$ from $U_n$. \cite{paninski2008coincidence} gave the optimal upper and lower bound in the case $p=1$ (with unknown constants) for large $n$; other results are new to my knowledge.

In \emph{learning}, we must output a distribution at $\ell_p$ distance at most $\epsilon$ from the given distribution, which has support size at most $n$. Optimal upper and lower bounds for learning in the cases $p=1$, $2$, and $\infty$ seem to the author to be all previously known as folklore (certainly for $\ell_1$ and $\ell_{\infty}$); others are new to my knowledge.}
\label{tab:results-summary}
\end{table}
%} %%%%%%%%%%%%%%%%%%%%%%%%% fbox

\subsection{Fixed bounds for large $n$ regimes}
A primary theme of the results is the intuition behind $\ell_p$ testing and learning in the case where the support size $n$ is large. In $\ell_p$ spaces for $p > 1$, we can achieve upper bounds for testing and learning that are independent of $n$.
%For for $p \geq 2$, for both problems $\bigTheta{\frac{1}{\epsilon^2}}$ samples are always sufficient (although for uniformity testing we can often do even better). For $1 < p \leq 2$, we observe the following behavior.
\begin{itheorem} For a fixed $p > 1$, let $q$ be the H\"{o}lder conjugate\footnote{Note that $1$ and $\infty$ are considered conjugates. This paper will also use math with infinity, so for instance, when $q=\infty$, $n^{1/q} = 1$ and it is never the case that $n \leq \frac{1}{\epsilon^q}$.} of $p$ with $\frac{1}{p}+\frac{1}{q}=1$. Let $n^* = 1/\epsilon^q$. Then $\bigO{\max\left\{\sqrt{n^*} ~,~ \frac{1}{\epsilon^2}\right\}}$ samples are sufficient for testing uniformity and $\bigO{\max\left\{n^* ~,~ \frac{1}{\epsilon^2}\right\}}$ are sufficient for learning. Furthermore, for $1 < p \leq 2$, when the support size $n$ exceeds $n^*$, then $\bigTheta{\sqrt{n^*}}$ and $\bigTheta{n^*}$ respectively are necessary and sufficient.
\end{itheorem}
Intuitively, particularly for $1 < p \leq 2$, we can separate into ``large $n$'' and ``small $n$'' regimes\footnote{For $p \geq 2$, this separation still makes sense in certain ways (see Observations \ref{obs:thin} and \ref{obs:approx-support} below) but does not appear in the sample complexity bounds in this paper.}, where the divider is $n^* = \frac{1}{\epsilon^q}$. In the small $n$ regime, tight bounds depend on $n$, but in the large $n$ regime where $n \geq n^*$, the number of samples is $\bigTheta{n^*}$ for learning and $\bigTheta{\sqrt{n^*}}$ for uniformity testing. This suggests the intuition that, in $\ell_p$ space with tolerance $\epsilon$, distributions' ``apparent'' support sizes are bounded by $n^* = \frac{1}{\epsilon^q}$. We next make two observations that align with this perspective, for purposes of intuition.
%Recall that we view a distribution $A$ on support size $n$ as a vector in $\mathbb{R}^n$ with probability $A_i$ on coordinate $i$.

\begin{restatable}{observation}{obsthin}
\label{obs:thin}
Let $1 < p$ and $q=\frac{p}{p-1}$. If the distribution $A$ is ``thin'' in that $\max_i A_i \leq \epsilon^{q}$, then $\|A\|_p \leq \epsilon$. In particular, if both distributions $A$ and $B$ are thin, then even if they are completely disjoint,
 \[ \|A-B\|_p \leq \|A\|_p + \|B\|_p \leq 2\epsilon. \]
\end{restatable}
\begin{proof} The claim holds immediately for $p=\infty$. For $1 < p < \infty$, by convexity, since $\sum_i A_i = 1$ and $\max_i A_i \leq \epsilon^q$, we have that $\|A\|_p^p = \sum_i A_i^p$ is maximized with as few nonzero entries as possible, each at its maximum value $\epsilon^q$. This extreme example is simply the uniform distribution on $n=\frac{1}{\epsilon^q}$, when $\|A\|_p^p = n\left(\frac{1}{n}\right)^p = \frac{1}{n^{p-1}} = \epsilon$. The rest is the triangle inequality.
\end{proof}
One takeaway from Observation \ref{obs:thin} is that if we are interested in an $\ell_p$ error tolerance of $\bigTheta{\epsilon}$, then any sufficiently ``thin'' distribution may almost as well be the uniform distribution on support size $\frac{1}{\epsilon^q}$. This perspective is reinforced by Observation \ref{obs:approx-support}, which says that under the same circumstances, any distribution may almost as well be ``discretized'' into $\frac{1}{\epsilon^q}$ chunks of weight $\epsilon^q$ each.

\begin{observation}
\label{obs:approx-support}
Fixing $1 < p$, for any distribution $A$, there is a distribution $B$ whose probabilities are integer multiples of $\frac{1}{\epsilon^q}$ such that $\|A-B\|_p \leq 2\epsilon$. In particular, $B$'s support size is at most $\frac{1}{\epsilon^q}$.
\end{observation}
\begin{proof}
We can always choose $B$ such that, on each coordinate $i$, $|A_i-B_i| \leq \frac{1}{\epsilon^q}$. (To see this, obtain the vector $B'$ by rounding each coordinate of $A$ up to the nearest integer multiple of $\epsilon^q$, and obtain $B''$ by rounding each coordinate down. $\|B'\|_1 \geq 1 \geq \|B''\|_1$, so we can obtain a true probability distribution by taking some coordinates from $B'$ and some from $B''$.) But this just says that the vector $A-B$ is ``thin'' in the sense of Observation \ref{obs:thin}. The same argument goes through here (even though $A-B$ is not a probability distribution): Since $\max_i |A_i-B_i| \leq \epsilon^q$ and $\sum_i |A_i - B_i| \leq 2$, by convexity $\|A-B\|_p$ is maximized when it has dimension $\frac{2}{\epsilon^q}$ and each entry $|A_i-B_i| = \epsilon^q$, so we get $\|A-B\|_p \leq 2\epsilon$.
\end{proof}

\subsection{Testing uniformity: biased coins and die} Given a coin, is it fair or $\epsilon$-far from fair? It is well-known that $\bigOmega{\frac{1}{\epsilon^2}}$ independent flips of the coin are necessary to make a determination with confidence. One might naturally assume that deciding if a $6$-sided die is fair or $\epsilon$-far from fair would only be more difficult, requiring more rolls, and one would be correct --- if the measure of ``$\epsilon$-far'' is $\ell_1$ distance. Indeed, it is known~\cite{paninski2008coincidence} that $\bigTheta{\frac{\sqrt{n}}{\epsilon^2}}$ rolls of an $n$-sided die are necessary if the auditor's distance measure is $\ell_1$.

But what about other measures, say, if the auditor wishes to test whether any one side of the die is $\epsilon$ more likely to come up than it should be? For this $\ell_{\infty}$ question, it turns out that \emph{fewer} rolls of the die are required than flips of the coin; specifically, we show that $\bigTheta{\frac{\ln n}{n \epsilon^2}}$ are necessary and sufficient, in a small $n$ regime (specifically, $\bigTheta{\frac{n}{\ln(n)}} \leq \frac{1}{\epsilon}$). Once $n$ becomes large enough, only $\bigTheta{\frac{1}{\epsilon}}$ samples are necessary and sufficient.

Briefly, the intuition behind this result in the $\ell_{\infty}$ case is as follows. When flipping a $2$-sided coin, both a fair coin and one that is $\epsilon$-biased will have many samples that are heads and many that are tails, making $\epsilon$ difficult to detect ($\frac{1}{\epsilon^2}$ flips are needed to overcome the variance of the process). On the other hand, imagine that we roll a die with $n = $one million faces, for which one particular face is $\epsilon = 0.01$ more likely to come up than it should be. Then after only $\bigTheta{\frac{1}{\epsilon}} = $ a few hundred rolls of the die, we expect to see this face come up multiple times. These multiple-occurrences or ``collisions'' are vastly less likely if the die is fair, so we can distinguish the biased and uniform cases.

So when the support is small, the variance of the uniform distribution can mask bias; but this fails to happen when the support size is large, making it easier to test uniformity over larger supports. These intuitions extend smoothly to the $\ell_p$ metrics below $p=\infty$: First, to be $\epsilon$-far from uniform on a large set, it must be the case that the distribution has ``heavy'' elements; and second, these heavy elements cause many more collisions than the uniform distribution, making them easier to detect than when the support is small.
However, this intuition only extends ``down'' to certain values of $p$.
\begin{itheorem} For $1 \leq p \leq 2$, for $n \leq n^* = \frac{1}{\epsilon^q}$, the sample complexity of testing uniformity is $\bigTheta{\frac{\sqrt{n}}{(n^{1/q}\epsilon)^2}}$. For $1 \leq p < \frac{4}{3}$, this is increasing in the support size $n$, and for $\frac{4}{3} < p \leq 2$, this is decreasing in $n$. For $p = \frac{4}{3}$, the sample complexity is $\bigTheta{\frac{1}{\epsilon^2}}$ for every value of $n$.
\end{itheorem}
Figure \ref{figure:test-unif} illustrates these bounds for different values of $p$, including the phase transition at $p=\frac{4}{3}$.

\begin{figure}
\includegraphics[width=1.08\columnwidth]{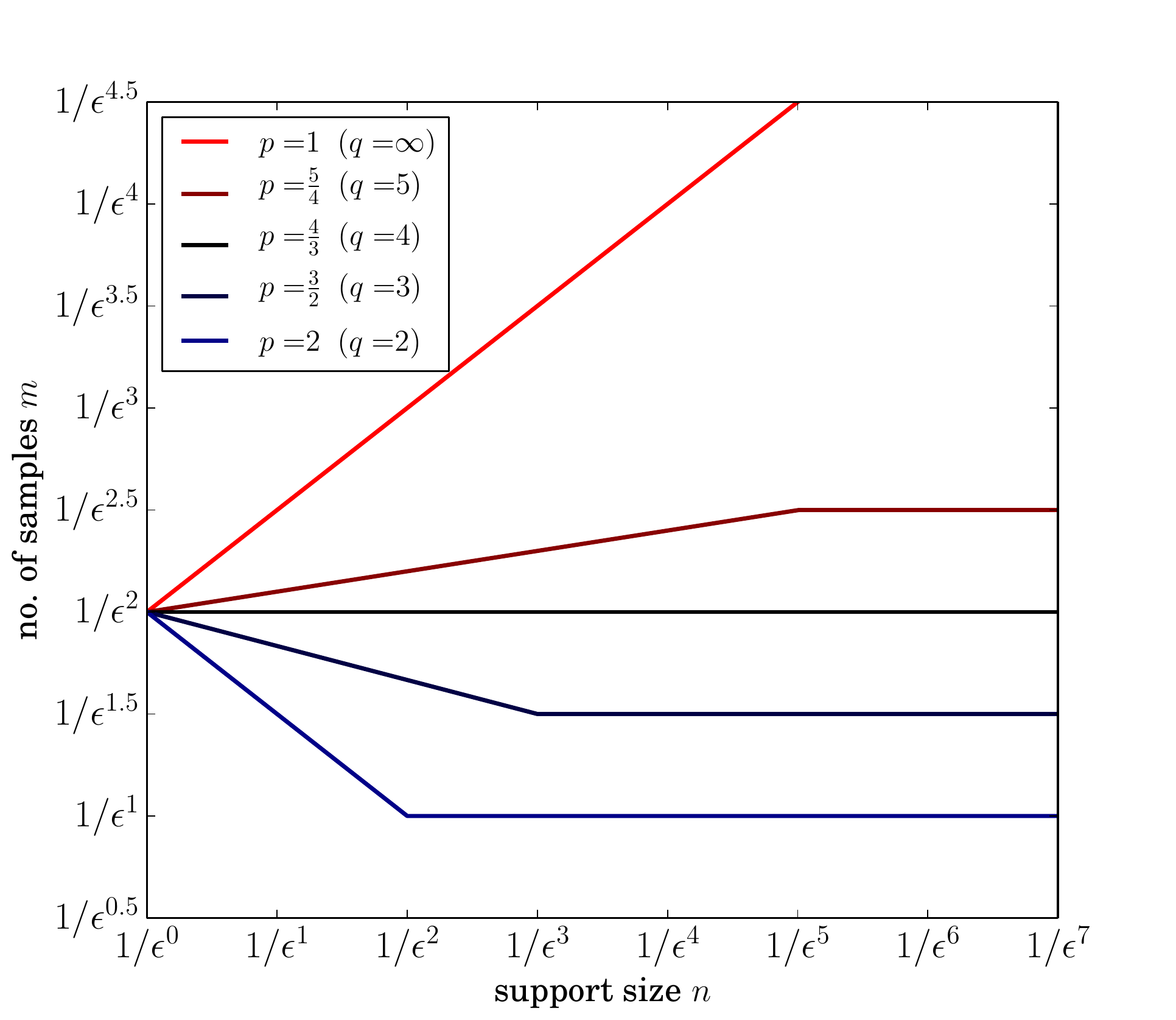}
\caption{Samples (necessary and sufficient, up to constant factors) for testing uniformity with a fixed $\ell_p$ tolerance $\epsilon$. On the horizontal axis is the support size $n$ of the uniform distribution, and on the vertical axis is the corresponding number of samples required to test uniformity. The function plotted is $\frac{\sqrt{n}}{(n^{1/q}\epsilon)^2}$ for $n \leq \frac{1}{\epsilon^q}$ and $\sqrt{\frac{1}{\epsilon^q}}$ for $n \geq \frac{1}{\epsilon^q}$, for various choices of $p$ and corresponding $q = \frac{p}{p-1}$. There is a phase transition at $p=\frac{4}{3}$: For $p < \frac{4}{3}$, the bound is initially increasing in $n$; for $p > \frac{4}{3}$, the bound is initially decreasing in $n$. For all $p$ except $p=1$, the number of necessary samples is constant for $n \geq \frac{1}{\epsilon^q}$. Note the log-log scale.}
\label{figure:test-unif}
\end{figure}

\vfill\eject
\newcommand{\uniftesthdr}{$1 \leq p \leq 2$}
\section{Uniformity Testing for \protect\uniftesthdr}
Recall the definition of uniformity testing: given i.i.d. samples from a distribution $A$, we must satisfy the following. If $A = U_n$, the uniform distribution on support size $n$, then with probability at least $1-\delta$, output ``uniform''. If $\|A-U_n\|_p \geq \epsilon$, then with probability at least $1-\delta$, output ``not uniform''.

\begin{algorithm}
\caption{Uniformity Tester}
\label{alg:test-unif}
 \begin{algorithmic}
  \State \emph{On input $p, n, \epsilon$, and failure probability $\delta$:}
  \State Choose $m$ to be ``sufficient'' for $p,n,\epsilon,\delta$ according to proven bounds.
  \State Draw $m$ samples.
  \State Let $C$ be the number of collisions: \\ \hspace{12pt} $C = \sum_{1\leq j < k \leq m} \1[\text{$j$th sample $=$ $k$th sample}]$.
  \State Let $T$ be the threshold: $T = {m\choose 2}\frac{1}{n} + \sqrt{\frac{1}{\delta}{m\choose 2}\frac{1}{n}}$.
  \State If $C \leq T$, output ``uniform''.
  \State If $C > T$, output ``not uniform''.
 \end{algorithmic}
\end{algorithm}
The upper bounds for $1 \leq p \leq 2$ rely on a very simple algorithm, Algorithm \ref{alg:test-unif}, and straightforward (if slightly delicate) argument. We count the number of \emph{collisions}: Pairs of samples drawn that are of the same coordinate. (Thus, if $m$ samples are drawn, there are up to ${m\choose 2}$ possible collisions.) The number of collisions $C$ has the following properties.\footnote{A possibly interesting generalization: The expected number of $k$-way collisions, for any $k=2,3,\dots$, is equal to ${m\choose k}\|A\|_k^k$. To prove it, consider the probability that each $k$-sized subset is such a collision (\emph{i.e.} all $k$ are of the same coordinate), and use linearity of expectation over the ${m\choose k}$ subsets.}
\begin{restatable}{lemma}{lemmacollisions}
\label{lemma:collisions}
On distribution $A$, the number of collisions $C$ satisfies:
\begin{enumerate}
 \item The expectation is \\
       $\mu_A = {m \choose 2}\|A\|_2^2 = {m\choose 2}\left(\frac{1}{n} + \|A-U\|_2^2\right)$.
 \item The variance is \\
       $Var(C) = {m \choose 2}\left(\|A\|_2^2-\|A\|_2^4\right) + 6{m \choose 3}\left(\|A\|_3^3 - \|A\|_2^4\right)$.
\end{enumerate}
\end{restatable}
Thus, the $\ell_2$ distance to uniform, $\|A-U\|_2$, intuitively controls the number of collisions we expect to see, with a minimum when $A=U$. This is why Algorithm \ref{alg:test-unif} simply declares the distribution nonuniform if the number of collisions exceeds a threshold.

\begin{restatable}{theorem}{theoremtestunifupper}
\label{theorem:test-unif-upper}
For uniformity testing with $1 \leq p \leq 2$, it suffices to run Algorithm \ref{alg:test-unif} while drawing the following number of samples:
 \[ m = \frac{9}{\delta} \begin{cases} \frac{\sqrt{n}}{\left(\epsilon n^{1/q}\right)^2}  &  n \leq \frac{1}{\epsilon^q}  \\[1.5em]
                                       \frac{1}{2} \sqrt{\left(\frac{2}{\epsilon}\right)^q}  &  n \geq \frac{1}{\epsilon^q} . \end{cases} \]
\end{restatable}
The proof of Theorem \ref{theorem:test-unif-upper} uses Chebyshev's inequality to bound the probability that $C$ is far from its expectation in terms of $Var(C)$, for both the case where $A=U_n$ and $\|A-U_n\|_p \geq \epsilon$. It focuses on a careful analysis of the variance of the number of collisions, to show that, for $m$ sufficiently large, the variance is small. For $1 \leq p \leq 2$, the dominant term eventually falls into one of two cases, which correspond directly to  ``large $n$'' ($n \geq \frac{1}{\epsilon^q}$) and ``small $n$'' ($n \leq \frac{1}{\epsilon^q}$).

Collisions, also called ``coincidences'', have been implicitly, but not explicitly, used to test uniformity for the $\ell_1$ case by Paninski~\cite{paninski2008coincidence}. Rather than directly testing the number of collisions, that paper tested the number of coordinates that were sampled exactly once. That tester is designed for the regime where $n > m$. Collisions have also been used for similar testing problems in \cite{goldreich2000testing, batu2013testing}. One interesting note is that $T$ is defined in terms of $m$, so that no matter how $m$ is chosen, if $A=U$ then the algorithm outputs ``uniform'' with probability $1-\delta$.

We also note that, if very high confidence is desired, a logarithmic dependence on $\delta$ is achievable by repeatedly running Algorithm \ref{alg:test-unif} for a fixed failure probability and taking a majority vote. The constants in the Theorem \ref{theorem:test-unif-upper-logdelta} are chosen to optimize the number of samples.
\begin{restatable}{theorem}{theoremtestunifupperlogdelta}
\label{theorem:test-unif-upper-logdelta}
For uniformity testing with $1 \leq p \leq 2$, it suffices to run Algorithm \ref{alg:test-unif} $160\ln(1/\delta)/9$ times, each with a fixed failure probability $0.2$, and output according to a majority vote; thus drawing a total number of samples
 \[ m = 800\ln(1/\delta) \begin{cases} \frac{\sqrt{n}}{\left(\epsilon n^{1/q}\right)^2}  &  n \leq \frac{1}{\epsilon^q}  \\[1.5em]
                                       \frac{1}{2} \sqrt{\left(\frac{2}{\epsilon}\right)^q}  &  n \geq \frac{1}{\epsilon^q} . \end{cases} \]
This improves on Theorem \ref{theorem:test-unif-upper} when the failure probability $\delta \leq 0.002$ or so.
\end{restatable}

The following lower bound shows that Algorithm \ref{alg:test-unif} is optimal for all $1 \leq p \leq 2$, $n$, and $\epsilon$, up to a constant factor depending on $p$ and the failure probability $\delta$.
\begin{restatable}{theorem}{theoremtestuniflower}
\label{theorem:test-unif-lower}
For uniformity testing with $1 \leq p \leq 2$, it is necessary to draw the following number of samples:
 \[ m = \begin{cases} \sqrt{\ln\left(1+(1-2\delta)^2\right)} \frac{\sqrt{n}}{\left(\epsilon n^{1/q}\right)^2}  & n \leq \frac{1}{\epsilon^q}   \\[1em]
                      \sqrt{2(1-2\delta)} \sqrt{\frac{1}{(2\epsilon)^q}}  & n \geq \frac{1}{\epsilon^q} . \end{cases} \]
\end{restatable}
In the large-$n$ regime, the lower bound can be proven simply. We pick randomly from a set of nonuniform distributions $A$ where, if not enough samples are drawn, then the probability of \emph{any} collision is very low. But without collisions, the input is equally likely to come from $U_n$ or from one of the nonuniform $A$s, so no algorithm can distinguish these cases.

In the small-$n$ regime, the order-optimal lower bound follows from the $\ell_1$ lower bound of Paninski~\cite{paninski2008coincidence}, which does not give constants. We give a rewriting of this proof with two changes: We make small adaptations to fit general $\ell_p$ metrics, and we obtain the constant factor. The idea behind the proof of \cite{paninski2008coincidence} is to again pick randomly from a family of distributions that are close to uniform. It is shown that any algorithm's probability of success is bounded in terms of the distance from the distribution of the resulting samples to that of samples drawn from $U_n$.

\vfill\eject
\newcommand{\uniftestptwohdr}{$p > 2$}
\section{Uniformity Testing for \protect\uniftestptwohdr}
This paper fails to characterize the sample complexity of uniformity testing in the $p > 2$ regime, except for the case $p = \infty$ in which the bounds are tight. However, the remaining gap is relatively small.

First, we note that Algorithm $1$ can be slightly adapted for use for all $p > 2$, giving an upper bound on the number of samples required. The reason is that, by an $\ell_p$-norm inequality, whenever $\|A-U\|_p \geq \epsilon$, we also have $\|A-U\|_2 \geq \epsilon$. So an $\ell_2$ tester is also an $\ell_p$ tester for $p \geq 2$. This observation proves the following theorem.
\begin{theorem} \label{theorem:test-unif-upper-use2}
For uniformity testing with any $p > 2$, it suffices to run Algorithm \ref{alg:test-unif} while drawing the number of samples for $p=2$ from Theorem \ref{theorem:test-unif-upper}, namely
  \[ m = \frac{9}{\delta}\begin{cases} \frac{1}{\sqrt{n}\epsilon^2}  & n \leq \frac{1}{\epsilon^2}  \\[1em]
                                       \frac{1}{\epsilon}            & n \geq \frac{1}{\epsilon^2} . \end{cases} \]
(A logarithmic dependence on $\delta$ is also possible as in Theorem \ref{theorem:test-unif-upper-logdelta}.) 
\end{theorem}
\begin{proof}
If $A=U$, then by the guarantee of Algorithm \ref{alg:test-unif}, with probability $1-\delta$ it outputs ``uniform''. If $\|A-U\|_p \geq \epsilon$, then $\|A-U\|_2 \geq \epsilon$: It is a property of $\ell_p$ norms that $\|V\|_2 \geq \|V\|_p$ for all vectors $V$ when $p \geq 2$. Then, by the guarantee of Algorithm \ref{alg:test-unif}, with probability $1-\delta$ it outputs ``not uniform''.
\end{proof}
The same reasoning, but in the opposite direction, says that a lower bound for the $\ell_{\infty}$ case gives a lower bound for all $p < \infty$. Thus, by proving a lower bound for $\ell_{\infty}$ distance, we obtain the following theorem.
\begin{theorem} \label{theorem:test-unif-lower-useinf}
For uniformity testing with any $p$, it is necessary to draw the following number of samples:
  \[ m = \begin{cases} \frac{1}{2} \frac{\ln\left(1 + n(1-2\delta)^2\right)}{n \epsilon^2}  &  \text{for all $n$}  \\[1em]
                       \frac{1-2\delta}{2} \frac{1}{\epsilon}                             &  n \geq \frac{1}{\epsilon} . \end{cases} \]
We find that the first bound is larger (better) for $\bigTheta{\frac{n}{\ln(n)}} \leq \frac{1}{\epsilon}$, and the second is better for all larger $n$.
\end{theorem}
\begin{proof}
In the
\ifnum\fullversion=1
appendix (Theorems \ref{theorem:test-unif-lower-infty-eps} and \ref{theorem:test-unif-lower-infty-n}),
\else
full version of the paper,
\fi
it is proven that this is a lower-bound on the number of samples for the case $p=\infty$. By the $p$-norm inequality mentioned above, for any $p \leq \infty$ and any vector $V$, $\|V\|_p \geq \|V\|_{\infty}$. In particular, suppose we had an $\ell_p$ testing algorithm. When the sampling distribution $A = U_n$, then by the guarantee of the $\ell_p$ tester it is correct with probability at least $1-\delta$; when $\|A-U_n\|_{\infty} \geq \epsilon$, we must have $\|A-U_n\|_p \geq \epsilon$ and so again by the guarantee of the $\ell_p$ tester it is correct with probability $1-\delta$. Thus the lower bound for $\ell_{\infty}$ holds for any $\ell_p$ algorithm as well.
\end{proof}
The lower bound for $\ell_{\infty}$ distance is proven by again splitting into the large and small $n$ cases. In the large $n$ case, we can simply consider the distribution
 \[ A^* = \left(\frac{1}{n} + \epsilon, ~ \frac{1}{n} - \frac{\epsilon}{n-1}, ~ \dots, ~ \frac{1}{n} - \frac{\epsilon}{n-1}\right) . \]
If $m$ is too small, then the algorithm probably does not draw any sample of the first coordinate; but conditioned on this, $A^*$ is indistinguishable from uniform (since it is uniform on the remaining coordinates).

In the small $n$ case, we adapt the general approach of \cite{paninski2008coincidence} that was used to prove tight lower bounds for the case $p \leq 2$. We consider choosing a random permutation of $A^*$ and then drawing $m$ i.i.d. samples from this distribution. As before, we bound the success probability of any algorithm in terms of the distance between the distribution of these samples and that of the samples from $U_n$.

Comparing Theorems \ref{theorem:test-unif-upper-use2} and \ref{theorem:test-unif-lower-useinf}, we see a relatively small gap for the small $n$ regime for $2 < p < \infty$, which is left open. A natural conjecture is that the sample complexity will be $\frac{1}{\epsilon}$ for the regime $n \geq \frac{1}{\epsilon^q}$. For the small $n$ regime, it is not clear what to expect; perhaps $\frac{1}{n^{1/q} \epsilon^2}$. New techniques seem to be required, since neither the analysis of collisions as in the case $p \leq 2$, nor the analysis of the single most different coordinate, as we will see for the $p=\infty$ case below, seems appropriate or tight for the case $2 < p < \infty$. \\

\newcommand{\myalpha}[1]{\frac{1}{#1}\left(1 + \frac{\ln(2#1)}{\ln(1/\delta)}\right)}
\textbf{A better $\ell_{\infty}$ tester.} For the $\ell_{\infty}$ case, the $\ell_2$ tester is optimal in the regime where $n \geq \frac{1}{\epsilon^2}$, as proven in Theorem \ref{theorem:test-unif-upper-use2}. For smaller $n$, a natural algorithm (albeit with some tricky specifics), Algorithm \ref{alg:test-unif-infty}, gives an upper bound that matches the lower bound up to constant factors. We first state this upper bound, then explain.

\begin{restatable}{theorem}{theoremtestinftysufficient}
\label{theorem:test-infty-sufficient}
For uniformity testing with $\ell_p$ distance, it suffices to run Algorithm \ref{alg:test-unif-infty} with the following number of samples:
 \[ m = \begin{cases} 23 \frac{\ln\left(\frac{2n}{\delta}\right)}{n \epsilon^2}  & \epsilon \leq 2\alpha(n)  \\[1em]
                      35 \frac{\ln\left(\frac{1}{\delta}\right)}{\epsilon}       & \epsilon > 2\alpha(n)  \end{cases} \]
where $\alpha(n) = \myalpha{n}$. In particular, for a fixed failure probability $\delta$, we have
 \[ \alpha(n) = \bigTheta{\frac{\ln(n)}{n}} . \]
\end{restatable}
To understand Algorithm \ref{alg:test-unif-infty}, consider separately the two regimes: $\bigTheta{\frac{n}{\ln(n)}} \leq \frac{1}{\epsilon}$ and otherwise. For details of the analysis, rather than phrasing the threshold in this way, we phrase it as $\epsilon \leq 2\alpha(n)$ where $\alpha(n) = \bigTheta{\frac{\ln(n)}{n}}$, but the actual form of $\alpha$ is more complicated because it depends on $\delta$.

In the first, smaller-$n$ regime, our approach will essentially be a Chernoff plus union bound. We will draw $m = \bigTheta{\frac{\ln(n)}{n \epsilon^2}}$ samples. Then Algorithm \ref{alg:test-unif-infty} simply checks for any coordinate with an ``outlier'' number of samples (either too many or too few). The proof of correctness is that, if the distribution is uniform, then by a Chernoff bound on each coordinate and union-bound over the coordinates, with high probability no coordinate has an ``outlier'' number of samples; on the other hand, if the distribution is non-uniform, then there is an ``outlier'' coordinate in terms of its probability and by a Chernoff bound this coordinate likely has an ``outlier'' number of samples.

In the second, larger-$n$ regime (where $\epsilon > 2\alpha(n)$), we will use the same approach, but first we will ``bucket'' the distribution into $\hat{n}$ groups where $\hat{n}$ is chosen such that $\epsilon = 2\alpha(\hat{n})$. In other words, no matter how large $n$ is, we choose $\hat{n}$ so that $\epsilon = \bigTheta{\frac{\ln(\hat{n})}{\hat{n}}}$ and treat each of the $\hat{n}$ groups as its own coordinate, counting the number of samples that group gets.

In this larger-$n$ regime, note that $\epsilon$ is large compared to the probability that the uniform distribution puts on each coordinate, or in fact on each group. So if $\|A-U\|_{\infty} \geq \epsilon$, then there is a ``heavy'' coordinate (and thus group containing it) that should get an outlier number of samples. We also need, by a Chernoff plus union bound, that under the uniform distribution, probably no group is an outlier. The key point of our choice of $\hat{n}$ is that it exactly balances this Chernoff plus union bound.
\begin{algorithm}
\caption{Uniformity Tester for $\ell_{\infty}$}
\label{alg:test-unif-infty}
 \begin{algorithmic}
  \State \emph{On input $n, \epsilon$, and failure probability $\delta$:}
  \State Choose $m$ to be ``sufficient'' for $n,\epsilon,\delta$ according to proven bounds.
  \State Draw $m$ samples.
  \State Let $\alpha(x) = \myalpha{x} = \bigTheta{\frac{\ln(x)}{x}}$.
  \If{$\epsilon \leq 2\alpha(n)$}
    \State Let $t = \sqrt{\frac{6m}{n} \ln\left(\frac{2n}{\delta}\right)}$.
    \State If, for all coordinates $i$, the number of samples $X_i \in \frac{m}{n} \pm t$, output ``uniform''.
    \State Otherwise, output ``not uniform''.
  \Else
    \State Let $\hat{n}$ satisfy $\epsilon = 2\alpha(\hat{n})$.
    \State Partition the coordinates into at most $2\lceil \hat{n} \rceil$ groups, each of size at most $\lfloor \frac{n}{\hat{n}} \rfloor$.
    \State For each group $j$, let $X_j$ be the total number of samples of coordinates in that group.
    \State Let $t = \sqrt{6m\epsilon \ln\left(\frac{1}{\delta}\right)}$.
    \State If there exists a group $j$ with $X_j \geq m\epsilon - t$, output ``not uniform''.
    \State Otherwise, output ``uniform''.
  \EndIf
 \end{algorithmic}
\end{algorithm}

\section{Distribution Learning}
Recall the definition of the learning problem: Given i.i.d. samples from a distribution $A$, we must output a distribution $\hat{A}$ satisfying that $\|A-\hat{A}\|_p \leq \epsilon$. This condition must be satisfied except with probability at most $\delta$.

\subsection{Upper Bounds}
Here, Algorithm \ref{alg:learner} is the natural/naive one: Let the probability of each coordinate be the frequency with which it is sampled.
\begin{algorithm}
\caption{Learner}
\label{alg:learner}
 \begin{algorithmic}
  \State \emph{On input $p, n, \epsilon$, and failure probability $\delta$:}
  \State Choose $m$ to be ``sufficient'' for $p,n,\epsilon,\delta$ according to proven bounds.
  \State Draw $m$ samples.
  \State Let $X_i$ be the number of samples drawn of each coordinate $i \in \{1,\dots,n\}$.
  \State Let each $\hat{A}_i = \frac{X_i}{m}$.
  \State Output $\hat{A}$.
 \end{algorithmic}
\end{algorithm}

The proofs of the upper bounds rely on an elegant proof approach which is apparently ``folklore'' or known for the $\ell_2$ setting, and was introduced to the author by Cl\'{e}ment Canonne\cite{clement2014private} who contributed it to this paper. The author and Canonne in collaboration extended the proof to general $\ell_p$ metrics in order to prove the bounds in this paper. Here, we give the theorem and proof for perhaps the most interesting and novel case, that for $1 < p \leq 2$, $\bigO{\frac{1}{\epsilon^q}}$ samples are sufficient independent of $n$. The other cases have a similar proof structure.
\begin{theorem} \label{theorem:learn-upper-non}
For $1 < p \leq 2$, to learn up to $\ell_p$ distance $\epsilon$ with failure probability $\delta$, it suffices to run Algorithm \ref{alg:learner} while drawing the following number of samples:
  \[ m = \left(\frac{3}{\delta}\right)^{\frac{1}{p-1}}\frac{1}{\epsilon^q} . \]
\end{theorem}
\begin{proof}
Let $X_i$ be the number of samples of coordinate $i$ and $\hat{A}_i = \frac{X_i}{m}$. Note that $X_i$ is distributed Binomially with $m$ independent trials of probability $A_i$ each. We have that
 \[ \E \|\hat{A}-A\|_p^p = \frac{1}{m^p} \sum_{i=1}^n \E \left|X_i - \E X_i\right|^p . \]
We will show that, for each $i$, $\E\left|X_i - \E X_i\right|^p \leq 3\E X_i$. This will complete the proof, as then
\begin{align*}
 \E\|\hat{A}-A\|_p^p
  &\leq \frac{1}{m^p} \sum_{i=1}^n 3\E X_i  \\
  &= \frac{1}{m^p} \sum_{i=1}^n 3m A_i  \\
  &= \frac{3}{m^{p-1}} ;
\end{align*}
and by Markov's Inequality,
\begin{align*}
 \Pr[\|\hat{A}-A\|_p^p \geq \epsilon^p]
  &\leq \frac{3}{m^{p-1} \epsilon^p} ,
\end{align*}
which for $m = \left(\frac{3}{\delta}\right)^{\frac{1}{p-1}}\frac{1}{\epsilon^q}$ is equal to $\delta$.

To show that $\E\left|X_i - \E X_i\right|^p \leq 3\E X_i$, fix any $i$ and consider a possible realization $x$ of $X_i$. If $|x - \E X_i| \geq 1$, then $|x - \E X_i|^p \leq |x - \E X_i|^2$. We can thus bound the contribution of all such terms by $\E|X_i - \E X_i|^2 = Var X_i$.

If, on the other hand, $|x - \E X_i| < 1$, then $|X_i - \E X_i|^p \leq |X_i - \E X_i|$; furthermore, at most two terms satisfy this condition, namely (letting $\beta := \lfloor \E X_i \rfloor$) $x = \beta$ and $x = \beta + 1$. These terms contribute a total of at most
\begin{align*}
  &\Pr[X_i = \beta]|\E X_i - \beta| + \Pr[X_i = \beta+1]|\beta + 1 - \E X_i| \\
 \leq &\E X_i + \Pr[X_i = \beta + 1] .
\end{align*}
Consider two cases. If $\E X_i \geq 1$, then the contribution is at most $\E X_i + 1 \leq 2\E X_i$. If $\E X_i < 1$, then $\beta + 1 = 1$, and by Markov's Inequality, $\Pr[X_i \geq 1] \leq \E X_i$, so the total contribution is again bounded by $2\E X_i$.

Thus, we have
\begin{align*}
 \E |X_i - \E X_i|^p
  &\leq Var X_i + 2\E X_i \\
  &\leq 3 \E X_i
\end{align*}
because $Var X_i = (1-A_i)\E X_i$.
\end{proof}

\vspace{12pt}
A slightly tighter analysis can be obtained by reducing to the $\ell_2$ algorithm, in which the above proof technique is ``tightest''. It produces the following theorem:

\begin{restatable}{theorem}{theoremlearnupperusetwo}
\label{theorem:learn-upper-use2}
For learning a discrete distribution with $1 \leq p \leq 2$, it suffices to run Algorithm \ref{alg:learner} with the following number of samples:
  \[ m = \frac{1}{\delta} \begin{cases} \frac{n}{\left(n^{1/q} \epsilon\right)^2}  & n \leq \left(\frac{2}{\epsilon}\right)^q  \\[1em]
                                        \frac{1}{4}\left(\frac{2}{\epsilon}\right)^q   & n \geq \left(\frac{2}{\epsilon}\right)^q . \end{cases} \]
With $p \geq 2$, it suffices to draw the sufficient number for $\ell_2$ learning, namely
  \[ m = \frac{1}{\delta} \frac{1}{\epsilon^2} . \]
\end{restatable}

In fact, $\|A-\hat{A}\|_p$ is tightly concentrated around its expectation, allowing a better asymptotic dependence on $\delta$ when high confidence is desired.
This idea is also folklore and not original to this paper. Here we apply it as follows. We must draw enough samples so that, first, the expectation of $\|\hat{A}-A\|_p$ is smaller than $\frac{\epsilon}{2}$; and second, we must draw enough so that, with probability $1-\delta$, $\|\hat{A}-A\|_p$ is no more than $\frac{\epsilon}{2}$ greater than its expectation.
It suffices to take the maximum of the number of samples that suffice for each condition to hold, resulting in the following bounds.
\begin{restatable}{theorem}{theoremlearnupperlogdelta}
\label{theorem:learn-upper-logdelta}
For learning a discrete distribution with $1 \leq p \leq 2$ and failure probability $\delta$, it suffices to run Algorithm \ref{alg:learner} with the following number of samples:
 \[ m = \max \left\{ \frac{2^{\frac{2}{p}+1} \ln(1/\delta)}{\epsilon^2} ~,~ M \right\} , \]
where
 \[ M = \begin{cases} 4\frac{n}{\left(n^{1/q} \epsilon\right)^2}  & n \leq \left(\frac{4}{\epsilon}\right)^q  \\[1em]
                                    \frac{1}{4} \left(\frac{4}{\epsilon}\right)^q  & n \geq \left(\frac{4}{\epsilon}\right)^q . \end{cases} \]
For $p \geq 2$, it suffices to use the sufficent number of samples for $\ell_2$ learning, namely 
 \[ m = \max\left\{ \frac{4\ln(1/\delta)}{\epsilon^2} ~,~ \frac{4}{\epsilon^2} \right\}. \]
\end{restatable}
In particular, for $\ell_1$ learning, it suffices to draw
 \[ m = \max\left\{\frac{8\ln(1/\delta)}{\epsilon^2} ~,~ \frac{4n}{\epsilon^2} \right\} . \]

\subsection{Lower bounds}
\begin{restatable}{theorem}{theoremlearnlower}
\label{theorem:learn-lower}
 To learn a discrete distribution in $\ell_p$ distance, the number of samples required for all $p,\delta$ is at least
 \[ m = \begin{cases} \bigOmega{\frac{1}{\epsilon^2}}  & 2 \leq p \leq \infty  \\[1em]
                      \bigOmega{\frac{1}{\epsilon^q}}  & 1 < p \leq 2, ~ n \geq \frac{1}{\epsilon^q} . \end{cases} \]
 For $1\leq p \leq 2$ and $n \leq \frac{1}{\epsilon^q}$, there is no $\gamma > 0$ such that
 \[ m = \bigO{\frac{n}{\left(n^{1/q}\epsilon\right)^{2-\gamma}}} \]
samples, up to a constant factor depending on $\delta$, suffice for all $\delta$.
\end{restatable}
As detailed in
\ifnum\fullversion=1
Appendix \ref{section:learn-lower},
\else
the full version,
\fi
these bounds can be proven from the folklore $\ell_1$ bound for the case $1 \leq p \leq 2$ (which seems to give a slightly tighter guarantee than the theorem statement); and the lower bound for $\ell_{\infty}$ uniformity testing gives the tight bound for $2 \leq p \leq \infty$. Finding it somewhat unsatisfying to reduce to the $\ell_1$ folklore result, we attempt an independent proof. This approach will give tight bounds up to (unspecified) constant factors for all $p$ and $\delta$ in the $1 < p \leq 2$, ``large $n$'' ($n \geq \frac{1}{\epsilon^q}$) regime. In the small $n$ regime, we will get bounds that look like $\frac{n}{(n^{1/q}\epsilon)^{2(1-\delta)}}$ instead of $\frac{n}{(n^{1/q}\epsilon)^2}$ (interpreted as in the above theorem). Thus, in this paper, the lower bound for this regime matches the upper bound in a weak sense; it would be nice if the below approach can be improved to yield a stronger statement.

\vspace{6pt}
We begin by defining the following game and proving the associated lemma:\\
\textbf{Distribution identification game:} The game is parameterized by maximum support size $n$, distance metric $\rho$, and tolerance $\epsilon$. First, a finite set $S$ of distributions is chosen with $\rho(A,B) > 2\epsilon$ for all $A,B \in S$. Every distribution in $S$ has support $\hat{n} \leq n$ (it will be useful to choose $\hat{n} \neq n$). The algorithm is given $S$. Second, a distribution $A \in S$ uniformly at random. Third, the algorithm is given $m$ i.i.d. samples from $A$. Fourth, the algorithm wins if it correctly guesses which $A \in S$ was chosen, and loses otherwise.

\vspace{6pt}
\begin{lemma} \label{lemma:learning-id-game} Any algorithm for learning to within distance $\epsilon$ using $m(n,p,\epsilon)$ samples with failure probability $\delta$ can be converted into an algorithm for distribution identification using $m(n,p,\epsilon)$ samples, with losing probability at most $\delta$. \end{lemma}

\begin{proof} Suppose the true oracle is $A \in S$. Run the learning algorithm, obtaining $\hat{A}$, and output the member $B$ of $S$ that minimizes $\rho(\hat{A},B)$ (where $\rho$ is the distance metric of the game; for us, it will be $\ell_p$ distance). With probability at least $1 - \delta$, by the guarantee of the learning algorithm, $\|\hat{A} - A\|_p \leq \epsilon$. When this occurs, we always output the correct answer, $A$: For any $B \neq A$ in $S$, by the triangle inequality $\|\hat{A} - B\|_p \geq \|B - A\| - \|\hat{A} - A\| > 2\epsilon - \epsilon = \epsilon$.
\end{proof}

The proofs of the lower bounds then proceed in the following fashion, at a high level:
\begin{enumerate}
 \item Construct a large set $S$ of distributions. For instance, for $1 \leq p \leq 2$, we have $|S| \approx \left(\frac{1}{(\hat{n})^{1/q}\epsilon}\right)^{\hat{n}}$. The main idea is to use a sphere-packing argument as with \emph{e.g.} the Gilbert-Varshamov bound in error-correcting codes. (In particular, the ``construction'' is not constructive; we merely prove that such a set exists.)
 \item Relate the probability of winning the game to the information obtained from the samples. Intuitively, we need a good ratio of the entropy of the samples, $\approx \hat{n} \log\left(\sqrt{\frac{m}{\hat{n}}}\right)$, to the entropy of the choice of distribution, $\log|S|$.
 \item Combine these steps. For instance, for $1 \leq p \leq 2$, we get that the probability of winning looks like $\left(\hat{n}^{1/q}\epsilon \sqrt{\frac{m}{\hat{n}}}\right)^{\hat{n}}$,
 implying that, for a constant probability of winning, we must pick $m \approx \frac{\hat{n}}{((\hat{n})^{1/q}\epsilon)^2}$.
 \item Choose $\hat{n} \leq n$. For $1 \leq p \leq 2$, in the small $n$ regime where $n \leq \frac{1}{\epsilon^q}$, the best choice turns out to be $\hat{n} = n$; in the large $n$ regime, the choice $\hat{n} = \frac{1}{\epsilon^q}$ turns out to be optimal and gives a lower bound $\bigTheta{\hat{n}}$ that is independent of $n$ for that range (since for any large enough $n$, we make the same choice of $\hat{n}$).
\end{enumerate}

\section{Prior and Future Work}
\label{section:related-work}
\subsection{Discussion of Prior Work}
The study of problems under $\ell_p$ metrics crops up in many areas of theoretical computer science and probability, as mentioned in the introduction. Similar in spirit to this paper is Berman et al 2014~\cite{berman2014testing}, which examined testing properties of real-valued functions such as monotonicity, Lipschitz constant, and convexity, all under various $\ell_p$ distances. Another case in which ``exotic'' metrics have been studied in connection with testing and learning is in Do et al 2011~\cite{do2011sublinear}, which studied the distance between and equality of two distributions under Earth Mover Distance.

For the problem of testing uniformity, Paninski 2008~\cite{paninski2008coincidence} examines the $\ell_1$ metric in the case of large-support distributions. The lower bound technique, which is slightly extended and utilized in this paper, establishes that $\bigOmega{\frac{\sqrt{n}}{\epsilon^2}}$ samples are necessary to test uniformity under the $\ell_1$ metric (with constants unknown). This lower bound holds for all support sizes $n$. The algorithm that gives the upper bound in that paper, a matching $m = \bigO{\frac{\sqrt{n}}{\epsilon^2}}$, holds for the case of very large support size $n$, namely $n > m$. This translates to $n = \bigOmega{\frac{1}{\epsilon^4}}$. The reason is that the algorithm counts the number of coordinates that are sampled exactly once; when $n > m$, this indirectly counts the number of collisions (more or less).

\cite{paninski2008coincidence} justifies a focus on $n > m$ because, for small $n$, one could prefer to just learn the distribution, which tells one whether it is uniform or not. However, depending on $\epsilon$, this paper shows that the savings can still be substantial: the number of samples required is on the order of $\frac{n}{\epsilon^2}$ to learn versus $\frac{\sqrt{n}}{\epsilon^2}$ to test uniformity using Algorithm \ref{alg:test-unif}. To the author's knowledge an order-optimal $\ell_1$ tester for all regimes may have previously been open. However, independently to this work, Diakonikolas et al 2015~\cite{diakonikolas2015testing} give an $\ell_2$ uniformity tester for the small-$n$ regime (which is optimal in that regime) and which implies an order-optimal $\ell_1$ tester for all parameters. They use a Poissonization and chi-squared-test approach.

More broadly, the idea of using collisions is common and also arises for related problems, \emph{e.g.} by \cite{goldreich2000testing} in a different context, and by Batu et al 2013~\cite{batu2013testing} for testing closeness of two given distributions in $\ell_1$ distance. This latter problem was resolved more tightly by Chan et al 2014~\cite{chan2014optimal} who established a $\bigTheta{\max\left\{\frac{n^{2/3}}{\epsilon^{4/3}}, ~~ \frac{\sqrt{n}}{\epsilon^2}\right\}}$ sample complexity. This problem may be a good candidate for future $\ell_p$ testing questions. It may be that the collision-based analysis can easily be adapted for general $\ell_p$ norms.

The case of learning a discrete distribution seems to the author to be mostly folklore. It is known that $\bigTheta{\frac{n}{\epsilon^2}}$ samples are necessary and sufficient in $\ell_1$ distance (as mentioned for instance in \cite{daskalakis2013learning}). It is also known via the ``DKW inequality''~\cite{dvoretzky1956asymptotic} that $\bigTheta{\frac{1}{\epsilon^2}}$ samples are sufficient in $\ell_{\infty}$ distance, with a matching lower bound coming from the biased coin setting (since learning must be at least as hard as distinguishing a $2$-sided coin from uniform). It is not clear to the author exactly what bounds would be considered ``known'' or ``folklore'' for the learning problem in $\ell_2$; perhaps the upper bound that $\bigO{\frac{1}{\epsilon^2}}$ samples are sufficient in $\ell_2$ distance is known. This work does provide a resolution to these questions, giving tight upper and lower bounds, as part of the general $\ell_p$ approach. But it should be noted that the results in at least these cases were already known and indeed the general upper-bound technique, introduced to the author by Cl\'{e}ment Canonne~\cite{clement2014private}, is not original here (possibly appearing in print for the first time).

\subsection{Bounds and Algorithms via Conversions}
As mentioned at times throughout the paper, conversions between $\ell_p$ norms can be used to convert algorithms from one case to another. In some cases this can give easy and tight bounds on the number of samples necessary and sufficient. The primary such inequality is Lemma \ref{lemma:pnorm-ineq}.
\begin{restatable}{lemma}{lemmapnormineq}
\label{lemma:pnorm-ineq}
 For $1 \leq p \leq s \leq \infty$, for all vectors $V \in \mathbb{R}^n$,
 \[ \frac{\|V\|_p}{n^{\frac{1}{p}-\frac{1}{s}}} ~ \leq ~ \|V\|_s ~ \leq ~ \|V\|_p . \]
\end{restatable}
For instance, suppose we have an $\ell_2$ learning algorithm so that, when it succeeds, we have $\|\hat{A}-A\|_2 \leq \alpha$. Then for $p > 2$, $\|\hat{A}-A\|_p \leq \|\hat{A}-A\|_2 \leq \alpha$, so we have an $\ell_p$ learner with the same guarantee. This also says that any lower bound for an $\ell_p$ learner, $p > 2$, immediately implies the same lower bound for $\ell_2$.

Meanwhile, for $p < 2$, $\|\hat{A}-A\|_p \leq \|\hat{A}-A\|_2 n^{\frac{1}{p}-\frac{1}{2}} \leq \alpha n^{\frac{1}{p}-\frac{1}{2}}$. This implies that, to get an $\ell_p$ learner for distance $\epsilon$, it suffices to use an $\ell_2$ learner for distance $\alpha = \epsilon n^{\frac{1}{2}-\frac{1}{p}} = \epsilon n^{1/q} / \sqrt{n}$. This can also be used to convert a lower bound for $\ell_p$, $p < 2$, into a lower bound for $\ell_2$ learners.

While these conversions can be useful especially for obtaining the tightest possible bounds, the techniques in this paper primarily focus on using a general technique that applies to all $\ell_p$ norms separately. However, it should be noted that applying these conversions to prior work can obtain some of the bounds in this paper (primarily for learning).

\subsection{Future Work}
An immediate direction from this paper is to close the gap on uniformity testing with $2 < p < \infty$, where $n$ is smaller than $\frac{1}{\epsilon^2}$. Although this case may be somewhat obscure or considered unimportant and although the gap is not large, it might require interesting new approaches.

A possibly-interesting problem is to solve the questions considered in this paper, uniformity testing and learning, when one is not given $n$, the support size. For uniformity testing, the question would be whether the distribution is $\epsilon$ far from every uniform distribution $U_n$, or whether it is equal to $U_n$ for some $n$. For each $p > 1$, these problems should be solvable without knowing $n$ by using the algorithms in this paper for the worst-case $n$ (note that, unlike the $p=1$ case, there is an $n$-independent maximum sample complexity). However, it seems possible to do better by attempting to learn or estimate the support size while samples are drawn and terminating when one is confident of one's answer.

A more general program in which this paper fits is to consider learning and testing problems under more ``exotic'' metrics than $\ell_1$, such as $\ell_p$, Earth Mover's distance \cite{do2011sublinear}, or others. Such work would benefit from finding motivating applications for such metrics. An immediate problem along these lines is testing whether two distributions are equal or $\epsilon$-far from each other in $\ell_p$ distance.

One direction suggested by the themes of this work is the testing and learning of ``thin'' distributions: those with small $\ell_{\infty}$ norm (each coordinate has small probability). For $p > 4/3$, we have seen that uniformity testing becomes easier over thinner distributions, where $n$ is larger. It also seems that we ought to be able to more quickly learn a thin distribution. At the extreme case, for $1 < p$, if $\max_i A_i \leq \epsilon^q$, then by Observation \ref{obs:thin}, we can learn $A$ to within distance $2\epsilon$ with zero samples by always outputting the uniform distribution on support size $\frac{1}{\epsilon^q}$. Thus, it may be interesting to consider learning (and perhaps other problems as well) as parameterized by the thinness of the distribution.

\section*{Acknowledgements}
The author thanks Cl\'{e}ment Canonne for discussions and contributions to this work.
Thanks to {\small \texttt{cstheory.stackexchange.com}}, via which the author first became interested in this problem.
Thanks to Leslie Valiant and Scott Linderman, teaching staff of Harvard CS 228, in which some of these results were obtained as a class project.
Finally, thanks to the organizers and speakers at the Workshop on Efficient Distribution Estimation at STOC 2014 for an interesting and informative introduction to and survey of the field.

%%%%%%%%%%%%%%%%%%%%%%%%%%%%%%%%%%%%%%%%%%%%%%

\vfill\eject

\bibliographystyle{abbrv}
\bibliography{references}

\ifnum\fullversion=1 %%%%%%%%%%%%%%

\vfill\eject
\break
\clearpage
\break

\appendix
The structure of the appendix matches the technical sections in the body of the paper.

\startcontents
\printcontents{1}{1}{}

\section{Preliminaries}
\label{section:preliminaries}
We consider discrete probability distributions of support size $n$, which will be represented as vectors $A \in \mathbb{R}^n$ where each entry $A_i \geq 0$ and $\sum_{i=1}^n A_i = 1$. We refer to $1,\dots,i,\dots,n$ as the coordinates.

$n$ will always be the support size of the distributions under consideration. $U_n$ will always refer to the uniform distribution on support size $n$, sometimes denoted $U$ where $n$ is evident from context. $m$ will always denote the number of i.i.d. samples drawn by an algorithm.

For $p \geq 1$, the $\ell_p$ norm of any vector $V \in \mathbb{R}^n$ is
 \[ \|V\|_p = \left(\sum_{i=1}^n |V_i|^p \right)^{1/p} . \]
The $\ell_{\infty}$ norm is
 \[ \|V\|_{\infty} = \max_{i=1,\dots,n} |V_i| . \]
For $1 \leq p \leq \infty$, the $\ell_p$ distance metric on $\mathbb{R}^n$ sets the distance between $V$ and $U$ to be $\|V-U\|_p$.

For a given $p$, $1 \leq p \leq \infty$, we let $q$ denote the H\"older conjugate of $p$: When $1 < p < \infty$, $q = \frac{p}{p-1}$ (and so $\frac{1}{p} + \frac{1}{q} = 1$); and $1$ and $\infty$ are conjugates of each other. We may use math with infinity. For instance, $\frac{1}{\infty}$ is treated as $0$. We may be slightly sloppy and, for instance, write $n \leq \frac{1}{\epsilon^q}$ when $q$ may be $\infty$, in which case (since $\epsilon < 1$) the expression is true for all $n$.

\paragraph{Goals} In all of the tasks considered in this paper, we are given $n \geq 2$ (the support size), $1 \leq p \leq \infty$ (specifying the distance metric), and $0 < \epsilon < 1$ (the ``tolerance''). We are given ``oracle access'' to a discrete probability distribution, meaning that we can specify a number $m$ and receive $m$ independent samples from the distribution.

We wish to determine the neccessary and sufficient number of i.i.d. samples to draw from oracle distributions in order to solve a given problem. The number of samples will always be denoted $m$; the goal is to determine the form of $m$ in terms of $n$, $p$, and $\epsilon$. The goal will be to return the correct (or a ``good enough'') answer with probability at least $1-\delta$ (we may call this the ``confidence''; $\delta$ is the ``failure probability''). For uniformity testing, $0 < \delta < 0.5$; for learning, $0 < \delta < 1$.

\subsection{Useful Facts and Intuition}
The first lemma is well-known and will be used in many places to relate the different norms of a vector. The second is used to relate norms independently of the support size.
%\begin{lemma} \label{lemma:pnorm-ineq} For $1 \leq p \leq s \leq \infty$, for all vectors $V \in \mathbb{R}^n$,
% \[ \frac{\|V\|_p}{n^{\frac{1}{p}-\frac{1}{s}}} ~ \leq ~ \|V\|_s ~ \leq ~ \|V\|_p . \]
%\end{lemma}
\lemmapnormineq*
\begin{proof} To show $\|V\|_s \leq \|V\|_p$: First, for $s=\infty$, we only need that
 \[ \left(\max_i |V_i|\right)^p \leq \sum_i |V_i|^p , \]
which is immediate. Now suppose $s < \infty$. Then we just need the following ratio to exceed $1$:\footnote{The idea of this trick was observed from http://math.stackexchange.com/questions/76016/is-p-norm-decreasing-in-p.}
\begin{align*}
 \left(\frac{\|V\|_p}{\|V\|_s}\right)^p
  &=    \sum_i \left(\frac{|V_i|}{\|V\|_s}\right)^p  \\
  &\geq \sum_i \left(\frac{|V_i|}{\|V\|_s}\right)^s  \\
  &=    1 .
\end{align*}
The inequality follows because, as already proven, for any $s$, $\|V\|_s \geq \max_i |V_i|$; so each term is at most $1$, and we have $s \geq p$, so the value decreases when raised to the $s$ rather than to the $p$.

It remains to show $\|V\|_p \leq n^{\frac{1}{p}-\frac{1}{s}}\|V\|_s$. Rewriting, we want to show
 \[ \frac{\|V\|_p}{n^{1/p}} \leq \frac{\|V\|_s}{n^{1/s}} . \]
If $s=\infty$, then we have
 \[ \left(\frac{\sum_i |V_i|^p}{n}\right)^{1/p} \leq \max_i |V_i| , \]
which follows because the maximum exceeds the average. For $s < \infty$, raise both sides to the $s$ power: We want to show 
 \[ \left(\frac{\sum_i |V_i|^p}{n}\right)^{\frac{s}{p}} \leq \frac{\sum_i |V_i|^s}{n} . \]
Since $\frac{s}{p} \geq 1$, the function $x \mapsto x^{\frac{s}{p}}$ is convex, and the above holds directly by Jensen's inequality.
\end{proof}

\begin{lemma} \label{lemma:pnorm-leq-2}
For any vector $V \in \mathbb{R}^n$ with $\|V\|_1 \leq c$:
\begin{enumerate}
 \item For $1 < p \leq 2$ with conjugate $q=\frac{p}{p-1}$,
  \[ \|V\|_p^q \leq c^{q-2} \|V\|_2^2 . \]
 \item For $2 \leq p \leq \infty$ with conjugate $q=\frac{p}{p-1}$,
  \[ \|V\|_p^q \geq c^{q-2} \|V\|_2^2 . \]
\end{enumerate}
\end{lemma}
\begin{proof}
We have
\begin{align}
 \|V\|_p^q
  &= \left(\sum_i |V_i|^p \right)^{\frac{1}{p-1}}  \nonumber \\
  &= \left(\|V\|_1 \sum_i \frac{|V_i|}{\|V\|_1} |V_i|^{p-1} \right)^{\frac{1}{p-1}}  \nonumber \\
  &= \left(\|V\|_1 \E |V_i|^{p-1} \right)^{\frac{1}{p-1}} \label{eqn:prelim-pnorm-2},
\end{align}
treating $\left(\frac{|V_1|}{\|V\|_1},\dots,\frac{|V_n|}{\|V\|_1}\right)$ as a probability distribution on $\{1,\dots,n\}$.
For the first claim of the lemma, by Jensen's inequality, since $p-1 \leq 1$ and the function $x \mapsto x^{p-1}$ is concave,
\begin{align*}
 \E |V_i|^{p-1}
  &\leq \left(\E |V_i|\right)^{p-1}  \\
  &=    \left(\frac{1}{\|V\|_1} \sum_i V_i^2 \right)^{p-1} ,
\end{align*}
which (plugging back into Equation \ref{eqn:prelim-pnorm-2}) gives
 \[ \|V\|_p^q \leq \|V\|_1^{\frac{2-p}{p-1}} ~ \|V\|_2^2 . \]
We have that $\frac{2-p}{p-1} = q-2$. And since for the first case $q-2 \geq 0$, the right side is maximized when $\|V\|_1 = c$.

For the second claim of the lemma, $p-1 \geq 1$, so by Jensen's inequality we get the exact same conclusion but with the inequality's direction reversed. (Note that in this case, $q-2 \leq 0$, so the right side is minimized when $\|V\|_1$ is at its maximum value $c$.)
\end{proof}
In particular, if $V$ is a probability distribution (so $\|V\|_1 = 1$), and $1 < p \leq 2$, then
 \[ \|V\|_p^q \leq \|V\|_2^2 \leq \|V\|_q^p. \]

\section{Uniformity Testing for \protect\uniftesthdr}
\label{section:test}

\subsection{Upper Bounds (sufficient)}

The upper-bound analysis focuses on the properties of $C$, the number of collisions, in Algorithm \ref{alg:test-unif}.
Recall that $C = \sum_{1\leq j < k \leq m} \1[\text{$j$th sample $=$ $k$th sample}]$; in other words, it is the number of pairs of samples that are of the same coordinate.
%\begin{lemma} \label{lemma:collisions}
%On distribution $A$, the number of collisions $C$ satisfies:
%\begin{enumerate}
% \item The expectation $\mu_A = {m \choose 2}\|A\|_2^2 = {m\choose 2}\left(\frac{1}{n} + \|A-U\|_2^2\right)$.
% \item The variance $Var(C) = {m \choose 2}\left(\|A\|_2^2-\|A\|_2^4\right) + 6{m \choose 3}\left(\|A\|_3^3 - \|A\|_2^4\right)$.
%\end{enumerate}
%\end{lemma}
\lemmacollisions*
\begin{proof}
(1) We have
 \begin{align*}
  \mu_A &= \E \sum_{1 \leq j \leq k \leq m} \1[S_j=S_k]  \\
        &= {m\choose 2} \Pr[S_j=S_k]  \\
        &= {m\choose 2} \sum_i \Pr[S_j = S_k = i]  \\
        &= {m\choose 2} \sum_i A_i^2 .
 \end{align*}
 Meanwhile,
 \begin{align*}
  \|A-U\|_2^2
   &= \sum_i \left(A_i - \frac{1}{n}\right)^2  \\
   &= \sum_i \left(A_i^2 - \frac{2}{n}A_i + \frac{1}{n^2}\right)  \\
   &= \sum_i A_i^2 ~~-~~ \frac{1}{n}
 \end{align*}
 using that $\sum_i A_i = 1$. \\

(2) Recall that we wrote $C$ as a sum of random variables $\1[S_j=S_k]$ for all pairs $j\neq k$. The variance of a sum of random variables is the sum, over all pairs of variables $\1[S_j=S_k]$ and $\1[S_x=S_y]$, of the covariances:
 \begin{align*}
  Var(C) &= \sum_{j\neq k} \sum_{x\neq y} Cov(\1[S_j=S_k],\1[S_x=S_y])  \\
         &= \sum_{j\neq k} \sum_{x\neq y} \Big( \E\left(\1[S_j=S_k]\1[S_x=S_y]\right) \\
         &~~~~~~~~~~~~~~  - \E\left(\1[S_j=S_k]\right)\E\left(\1[S_x=S_y]\right)\Big)  \\
         &= \sum_{j\neq k} \sum_{x\neq y} \Big( \Pr[S_j=S_k \text{ and } S_x=S_y]  \\
         &~~~~~~~~~~~~~~  - \Pr[S_j=S_k]\Pr[S_x=S_y] \Big).
 \end{align*}
If all four of $j,k,x,y$ are distinct, \emph{i.e.} the two pairs of samples have no samples in common, then the events $S_j=S_k$ and $S_x=S_y$ are independent, so all of these terms in the summation are zero. Otherwise, first note that the right summand is
 \begin{align*}
  \Pr[S_j=S_k]\Pr[S_x=S_y] &= \left(\sum_i A_i^2\right)^2  \\
                           &= \|A\|_2^4 .
 \end{align*}
 Now consider the case where the pairs are equal: $\{j,k\} = \{x,y\}$. This case holds for ${m \choose 2}$ choices of $\{j,k,x,y\}$ (namely, all possible pairs $j \neq k$), and when it holds,
 \begin{align*}
  \Pr[S_j=S_k \text{ and } S_x=S_y] &= \Pr[S_j = S_k]  \\
                                    &= \|A\|_2^2 .
 \end{align*}
 The final case is where the pairs have one index in common: $|\{j,k\} \cap \{x,y\}| = 3$. This case holds for all possible unequal triples of indices, ${m \choose 3}$ triples, and for each one it appears $6$ times in the sum: If $a < b < c$, we have (1)$j=a,k=b,x=b,y=c$; (2) $j=a,k=c,x=b,y=c$; (3) $j=a,k=b,x=a,y=c$, and the symmetric three cases with $(j,k)$ swapped with $(x,y)$. So, to reiterate, this case holds for $6{m\choose 3}$ terms in the sum. When it holds,
 \begin{align*}
  \Pr[S_j=S_k \text{ and } S_x=S_y] &= \Pr[S_j=S_k=S_x=S_y]  \\
                                    &= \Pr[\text{three samples are all equal}]  \\
                                    &= \sum_i A_i^3  \\
                                    &= \|A\|_3^3 .
 \end{align*}
 Putting it all together, we get that
 \[ Var(C) = {m \choose 2}\left(\|A\|_2^2 - \|A\|_2^4\right) + 6{m\choose 3}\left(\|A\|_3^3 - \|A\|_2^4\right) . \]
 (For a sanity check, we can notice that we got ${m \choose 2} + 6{m \choose 3}$ nonzero terms in the sum. Let us count the zero terms: the ones where $j,k,x,y$ are all distinct.\footnote{This is \emph{not} ${m\choose 4}$, because a given set of four distinct indices can appear as $j,k,x,y$ in $6$ different ways (one can check), giving $6{m\choose 4} = {m\choose 2}{m-2 \choose 2}$.} Thus, count all the ways we can first pick $j\neq k$, which is ${m \choose 2}$, times all the ways we can pick $x \neq y$ from the remaining $m-2$ indices, which is ${m-2 \choose 2}$. Thus, the number of zero terms is ${m \choose 2}{m-2 \choose 2}$. Now, to complete the sanity check, note that in total there are ${m \choose 2}^2$ terms in the sum, and we do have ${m\choose 2}{m-2 \choose 2} + {m\choose 2} + 6{m \choose 3} = {m\choose 2}^2$.)
\end{proof}

\theoremtestunifupper*

We give a proof sketch before giving the full proof.

\textit{Proof sketch.} Given Lemma \ref{lemma:collisions}, the proof is intuitively straightforward (if slightly tedious).
Recall that the threshold is
 \[ T = {m\choose 2}\frac{1}{n} + \sqrt{\frac{1}{\delta}{m\choose 2}\frac{1}{n}} . \]
We output ``uniform'' if and only if $C \leq T$.

$T$ was chosen to ``fit'' the expectation and variance of the collisions when the oracle $A$ is the uniform distribution.
In that case, the expected number of collisions is $\mu_A = {m\choose 2}\frac{1}{n}$ and the variance is $Var(C) \leq \mu_A$ (it turns out).
Thus, by Chebyshev, $\Pr[C \geq T] \leq \Pr[|C-\mu_A| \geq \sqrt{\mu_A/\delta}] \leq \delta Var(C)/\mu_A \leq \delta$.
This argument holds for all choices of $m$, since we chose $T$ depending on $m$.

If the oracle is some $A$ with $\|A-U\|_p \geq \epsilon$, then we again apply Chebyshev's inequality, looking to bound the $\Pr[C < T]$.
The variance is made up of several additive terms, and in different regimes different terms will dominate. Knowing the correct form of $m$ ``in advance'', and plugging it in, simplifies the case analysis somewhat and enables us to solve for a constant.

\begin{proof}
First, we prove that, if $A=U$, then with probability at least $1-\delta$, we output ``uniform''.
By Chebyshev's Inequality,
\begin{align*}
 \delta
  &\geq \Pr\left[|C - \mu_U| \geq \sqrt{Var(C)/\delta}\right]  \\
  &\geq \Pr\left[C \geq \mu_U + \sqrt{\mu_U/\delta}\right]  \\
  &=    \Pr\left[C \geq T\right].
\end{align*}
We used Lemma \ref{lemma:collisions}, the definition of $T$, and the observation that when drawing from the uniform distribution, $Var(C) \leq {m\choose 2}\|U\|_2^2 = \mu_U$, because $\|U\|_3^3 = \|U\|_2^4 = \frac{1}{n^2}$.
(Note that this proof works for any $m$, since the threshold is chosen as the ``correct'' function of $m$.
 The bound on $m$ is only needed for the next part of the proof.)

Next, and more involved, is the proof that, if $\|A-U\|_p \geq \epsilon$, then with probability at least $1-\delta$, we output ``different''.
Again, we will employ Chebyshev, this time to bound\footnote{Note this argument requires $\mu_A - T > 0$, which will turn out from the math below to be true if $m \geq \frac{\sqrt{6}}{\sqrt{n}\|A-U\|_2^2}$, and it will turn out that we always pick $m$ larger than this.}
\begin{align*}
 \Pr\left[C \leq T\right]
  &= \Pr\left[\mu_A - C \geq \mu_A - T\right]  \\
  &\leq \Pr\left[ |\mu_A - C| \geq \mu_A - T \right]  \\
  &\leq \frac{Var(C)}{(\mu_A - T)^2} .
\end{align*}
So we need to pick $m$ so that, when $\|A-U\|_p \geq \epsilon$,
\begin{equation}
 Var(C) \leq \delta \left(\mu_A - T\right)^2 . \label{eqn:dist-unif-upper-key}
\end{equation}
Recall that
$ \mu_A = {m\choose 2}\left(\frac{1}{n} + \|A-U\|_2^2\right) $
and
$ T = {m\choose 2}\frac{1}{n} + \sqrt{\frac{1}{\delta}{m\choose 2}\frac{1}{n}} . $
Thus,
\begin{align*}
 \mu_A - T &= {m\choose 2}\|A-U\|_2^2 - \sqrt{{m\choose 2}\frac{1}{\delta n}} ,
\end{align*}
so the right side of Inequality \ref{eqn:dist-unif-upper-key} is
\begin{align}
 &\delta \left(\mu_A - T\right)^2  \nonumber \\
  &= \delta {m\choose 2}^2\|A-U\|_2^4 - 2{m\choose 2}^{3/2}\|A-U\|_2^2\sqrt{\frac{\delta}{n}} + {m\choose 2}\frac{1}{n} . \label{eqn:dist-unif-upper-muAT}
\end{align}
Meanwhile, we claim that the left side satisfies the inequality
\begin{align}
 Var(C) \leq &{m\choose 2}\frac{1}{n} + \nonumber \\
             &{m\choose 2}\|A-U\|_2^2\left(1 + 2(m-2)\left(\frac{1}{n} + \|A-U\|_2\right)\right) . \label{eqn:dist-unif-upper-var}
\end{align}
We defer the proof of Inequality \ref{eqn:dist-unif-upper-var} and first show how it is used to prove the lemma.
Recall that the goal is to choose $m$ so that Inequality \ref{eqn:dist-unif-upper-key} holds. We can be assured that Inequality \ref{eqn:dist-unif-upper-key} holds if the right side of Inequality \ref{eqn:dist-unif-upper-var} is at most the right side of Equation \ref{eqn:dist-unif-upper-muAT}.
After subtracting ${m\choose 2}\frac{1}{n}$ from both sides and dividing both sides by ${m\choose 2}\|A-U\|_2^2$, this reduces to
\begin{align*}
 1 + 2(m-2)\left(\frac{1}{n} + \|A-U\|_2\right)
  &\leq \delta {m\choose 2}\|A-U\|_2^2 - 2\sqrt{\frac{{m\choose 2}\delta}{n}} .
\end{align*}
Apply on the right side that ${m\choose 2} \leq \frac{m^2}{2}$,\footnote{Justified because the right side is positive implies that this substitution increases it.} move the rightmost term to the other side, and divide through by $\delta \frac{m^2}{2} \|A-U\|_2^2$: it suffices that
\begin{align}
 &\frac{2\sqrt{2}}{m \sqrt{\delta n} \|A-U\|_2^2} + \frac{2}{\delta m^2 \|A-U\|_2^2}  \nonumber \\
 &+ \frac{4}{\delta n m \|A-U\|_2^2} + \frac{4}{\delta m \|A-U\|_2} ~~ \leq 1. \label{eqn:dist-unif-sum-var}
\end{align}
Now, suppose that $m$ satisfies
\begin{align}
 m \geq \frac{k}{\delta} \max\left\{\frac{1}{\sqrt{n}\|A-U\|_2^2}, ~\frac{1}{\|A-U\|_2}\right\} . \label{eqn:dist-unif-m-condition}
\end{align}
Then we get the requirement
 \[ \frac{2\sqrt{2\delta}}{k} + \frac{2\delta}{k^2} + \frac{4}{k\sqrt{n}} + \frac{4}{k} \leq 1, \]
which, since $\delta < 0.5$ and $n \geq 2$, we can check is satisfied for $k = 9$ (or actually $k \geq 8.940...$).

It remains to ensure that $m$ satisfies Inequality \ref{eqn:dist-unif-m-condition}, which is in terms of $\|A-U\|_2$; but we are given a guarantee of the form $\|A-U\|_p \geq \epsilon$.
For $p \leq 2$, since $\|A-U\|_p \geq \epsilon$, we have by Lemmas \ref{lemma:pnorm-ineq} and \ref{lemma:pnorm-leq-2} that
 \[ \|A-U\|_2 \geq \alpha := \max\left\{ \frac{\epsilon}{n^{\frac{1}{2}-\frac{1}{q}}}, ~~ \frac{\epsilon^{q/2}}{2^{\frac{q-2}{2}}} \right\} , \]
plugging in that $\|A-U\|_1 \leq 2$. For $n \leq \frac{1}{(2\epsilon)^q}$, the first term is larger, and we get that
 \[ m \geq \frac{9}{\delta} \max\left\{\frac{n^{\frac{1}{2}-\frac{2}{q}}}{\epsilon^2}, ~~ \frac{2^{\frac{q-2}{2}}}{\epsilon^{q/2}}\right\} \]
samples suffices.
This completes the proof, except to show Inequality \ref{eqn:dist-unif-upper-var} as promised.

To prove it, start by dropping the relatively insignificant first $\|A\|_2^4$ term:
 \begin{align*}
  Var(C) &\leq {m\choose 2}\left(\|A\|_2^2 + 2(m-2)\left(\|A\|_3^3 - \|A\|_2^4\right)\right)  \\
 \end{align*}
We will show that
\begin{align*}
 \|A\|_3^3 - \|A\|_2^4
  &\leq \|A\|_2^2 ~ \left(\frac{1}{n} + \|A-U\|_2\right) .
\end{align*}
One can check that this will complete the proof of Inequality \ref{eqn:dist-unif-upper-var}, by substituting and rearranging (also using that $\|A\|_2^2 = \frac{1}{n} + \|A-U\|_2^2$).

To show that $\|A\|_3^3 - \|A\|_2^4 \leq \|A\|_2\left(\frac{1}{n} + \|A-U\|_2\right)$, introduce the notation $\delta_i = A_i - \frac{1}{n}$. (This is unrelated to the failure probability.) Then with some rearranging (note that $\sum_i \delta_i = 0$),
\begin{align*}
 \|A\|_3^3 &= \sum_i \left(\frac{1}{n} + \delta_i\right)^3  \\
           &= \frac{1}{n^2} + \sum_i \delta_i^2\left(\frac{3}{n} + \delta_i\right)
\end{align*}
and
\begin{align*}
 \|A\|_2^4 &= \left(\frac{1}{n} + \sum_i \delta_i^2\right)^2  \\
           &= \frac{1}{n^2} + \sum_i \delta_i^2\left(\frac{2}{n} + \sum_j \delta_j^2\right) .
\end{align*}
Thus, the difference is at most (dropping the relatively insignificant $\sum_j \delta_j^2$ term)
\begin{align*}
 \|A\|_3^3 - \|A\|_2^4
  &\leq \sum_i \delta_i^2\left(\frac{1}{n} + \delta_i\right)  \\
  &= \|A-U\|_2^2\frac{1}{n} + \|A-U\|_3^3 .
\end{align*}
At this point, use the fact from Lemma \ref{lemma:pnorm-ineq} that $\|A-U\|_3 \leq \|A-U\|_2$ to get
 \[ \|A\|_3^3 - \|A\|_2^4 \leq \|A-U\|_2^2\left(\frac{1}{n} + \|A-U\|_2\right) . \]
\end{proof}

\theoremtestunifupperlogdelta*
\begin{proof}
Suppose we run Algorithm \ref{alg:test-unif} $k$ times, each with a fixed failure probability $\delta'$. The number of samples is $k$ times the number given in Theorem \ref{theorem:test-unif-upper} (with parameter $\delta'$). Each iteration is correct independently with probability at least $1-\delta'$, so the probability that the majority vote is incorrect is at most the probability that a Binomial of $k$ draws with probability $1-\delta'$ each has at most $k/2$ successes; by a Chernoff bound (\emph{e.g.} Mitzenmacher and Upfal~\cite{mitzenmacher2005probability}, Theorem 4.5),
\begin{align*}
 \Pr[ \text{\# successes} \leq k/2 ]
  &\leq \exp\left[ \frac{-\left((1-\delta')k - \frac{k}{2}\right)^2}{2(1-\delta')k} \right]  \\
  &= \exp\left[ \frac{-k\left(\frac{1}{2} - \delta'\right)^2}{2\left(1-\delta'\right)} \right] .
\end{align*}
Thus, it suffices to set
 \[ k = \ln\left(\frac{1}{\delta}\right)\left(\frac{2(1-\delta')}{\left(\frac{1}{2} - \delta'\right)^2}\right) . \]
(Technically there ought to be a ceiling function around this expression in order to make $k$ an integer.) This holds for any choice of $\delta' < 0.5$, but it is approximately minimized by $\delta' = 0.2$, when $k = \frac{160}{9}\ln(1/\delta)$. Each iteration requires the number of samples stated in Theorem \ref{theorem:test-unif-upper} with failure probability $\delta' = 0.2$, which completes the proof of the theorem.
\end{proof}

\subsection{Lower Bounds (necessary)}

\theoremtestuniflower*
\begin{proof} The proof will be given separately for the two separate cases by (respectively) Theorems \ref{theorem:test-unif-lower-smallp-smalln} and \ref{theorem:test-unif-lower-smallp-bign}.
\end{proof}

\begin{theorem} \label{theorem:test-unif-lower-smallp-bign}
For uniformity testing with $1 < p \leq 2$ and $n \geq \frac{1}{\epsilon^q}$, with failure probability $\delta$, it is necessary to draw at least the following number of samples:
 \[ m = \sqrt{2(1-2\delta)\frac{1}{(2\epsilon)^q}} . \]
\end{theorem}
\textit{Proof sketch.} We will construct a family of distributions, all of which are $\epsilon$-far from uniform. We will draw a member uniformly randomly from the family, and give the algorithm oracle access to it.
If the algorithm has failure probability at most $\delta$, then it outputs ``not uniform'' with probability at least $1-\delta$ on average over the choice of oracle (because it does so for every oracle in the family).

However, the algorithm must also say ``uniform'' with probability at least $1-\delta$ when given oracle access to $U$. The idea will be that, on both the uniform distribution and one chosen from the family, the probability of \emph{any} collision is very low. But, conditioned on no collisions, a randomly chosen member of the family is completely indistinguishable from uniform. So if the algorithm usually says ``uniform'' when the input has no collisions, then it is usually wrong when the oracle is drawn from our family; or vice versa.

\begin{proof}
Construct a family of distributions as follows. We will choose a particular value $\hat{n} \leq \frac{n}{2}$ (to be specified later). Pick $\hat{n}$ coordinates uniformly at random from the $n$ coordinates, and let each have probability $\frac{1}{\hat{n}}$. The remaining coordinates have probability zero.

We will need to confirm two properties: that $\|A-U\|_p \geq \epsilon$ for every $A$ in the family, and that the probability of any collision occurring is small. Toward the first property, we have that on each of the $\hat{n}$ nonzero coordinates, $|A_i - \frac{1}{n}| = \frac{1}{\hat{n}} - \frac{1}{n} \geq \frac{1}{2\hat{n}}$, using that $\frac{1}{n} \leq \frac{1}{2\hat{n}}$. Thus,
\begin{align}
 \|A-U\|_p^p &\geq \hat{n}\left(\frac{1}{2\hat{n}}\right)^p  \nonumber \\
             &= \frac{1}{2^p(\hat{n})^{p-1}} .  \label{eqn:dist-unif-lower-lpnorm}
\end{align}
So for the first property, $\ell_p$ distance $\epsilon$ from uniform, we must choose $\hat{n}$ so that Expression \ref{eqn:dist-unif-lower-lpnorm} is at least $\epsilon^p$.
For the property that the chance of a collision is small, we have by Markov's Inequality that for any $A$ in the family,
\begin{align}
  \Pr[C \geq 1] &\leq \E[C]  \nonumber \\
                &=    {m\choose 2}\|A\|_2^2  \nonumber \\
                &=    {m\choose 2}\hat{n}\left(\frac{1}{\hat{n}}\right)^2 \nonumber \\
                &\leq \frac{m^2}{2\hat{n}} . \label{eqn:dist-unif-lower-prob-col}
\end{align}
Now we choose $\hat{n} = \left(\frac{1}{2\epsilon}\right)^q$. Note that, if $n \geq \frac{1}{\epsilon^q}$, then $\hat{n} = \frac{n}{2^q} \leq \frac{n}{2}$.
For the first property, for any distribution $A$ in the family, by Inequality \ref{eqn:dist-unif-lower-lpnorm}, $\|A-U\|_p^p \geq \frac{(2\epsilon)^{q(p-1)}}{2^p} = \epsilon^p$.
For the second property, by Inequality \ref{eqn:dist-unif-lower-prob-col}, $\Pr[C \geq 1] \leq m^2 \left(2\epsilon\right)^q/2$, so if $m < \sqrt{2\frac{1-2\delta}{(2\epsilon)^q}}$, then
 \[ \Pr[C \geq 1] \leq 1-2\delta . \]
This shows that, if the oracle is drawn from the family, then the expected number of collisions, and thus probability of any collision, is less than $1-2\delta$ if $m$ is too small.
Meanwhile, if the oracle is the uniform distribution $U$, then the expected number of collisions is smaller (since $\|U\|_2^2 = \frac{1}{n} \leq \|A\|_2^2$).
So if $m$ is smaller than the bounds given, then for either scenario of oracle, the algorithm observes a collision with probability less than $1-2\delta$.

But if there are no collisions, then the input consists entirely of distinct samples and every such input is equally likely, under both the oracle being $U$ and under a distribution chosen uniformly from our family (by symmetry of the family).
Thus, conditioned on zero collisions, the probability $\gamma$ of the algorithm outputting ``uniform'' is equal when given oracle access to $U$ and when it is given oracle access to a uniformly chosen member of our family of distributions.
If $\gamma \leq \frac{1}{2}$, then the probability of correctness when given oracle access to $U$ is at most $\gamma\cdot \Pr[\text{no collisions}] + \Pr[\text{collisions}] \leq \frac{1}{2} + \frac{1}{2}\Pr[\text{collisions}] \leq \frac{1}{2} + \frac{1}{2}(1-2\delta) = 1-\delta$.
Conversely, if $\gamma \geq \frac{1}{2}$, then the probability of correctness when given oracle access to a member of the family is at most $(1-\gamma)\Pr[\text{no collisions}] + \Pr[\text{collisions}] \leq \frac{1}{2} + \frac{1}{2}\Pr[\text{collisions}] \leq 1-\delta$ again.
\end{proof}

\begin{theorem} \label{theorem:test-unif-lower-smallp-smalln}
For uniformity testing with $1 \leq p \leq 2$, if $n \leq \frac{1}{\epsilon^q}$, then it is necessary to draw the following number of samples:
  \[ m = \sqrt{\ln\left((1-2\delta)^2+1\right)} \frac{\sqrt{n}}{\left(\epsilon n^{1/q}\right)^2} . \]
\end{theorem}
\begin{proof}
We know from \cite{paninski2008coincidence} that, in $\ell_1$ norm, $\Omega\left(\frac{\sqrt{n}}{\epsilon^2}\right)$ samples are required. This result actually immediately implies the bound with an unknown constant, by a careful change of parameters, as follows. Suppose that $A$ satisfies $\|A-U\|_p \leq \epsilon$, for $1 \leq p \leq \infty$. Then by Lemma \ref{lemma:pnorm-ineq}, $\|A-U\|_1 \leq \epsilon n^{1-\frac{1}{p}} = \epsilon n^{1/q}$. So let $\alpha = \epsilon n^{1/q}$. Then since $\|A-U\|_1 \leq \alpha$, the number of samples required to distinguish $A$ from $U$ is on the order of
 \[ \frac{\sqrt{n}}{\alpha^2} = \frac{\sqrt{n}}{\left(n^{1/q}\epsilon\right)^2} . \]

Below, we chase through the construction and analysis (somewhat modified for clarity, it is hoped) of \cite{paninski2008coincidence}, adapted for the general case. The primary point of the exercise is to obtain the constant in the bound, which is not apparent in \cite{paninski2008coincidence}. 

So fix $1 \leq p \leq 2$. The plan is to construct a set of distributions and draw one uniformly at random, then draw $m$ i.i.d. samples from it. These samples are distributed in some particular way; let $\vec{Z}$ be their distribution (written as a length-$n^m$ vector, since there are $n^m$ possible outcomes). Let $\vec{U}$ be the distribution of the $m$ input samples when the oracle distribution is $U$; $\vec{U} = \left(\frac{1}{n^m},\dots,\frac{1}{n^m}\right)$ since every outcome of the $m$ samples is equally likely.

Suppose that the algorithm, which outputs either ``unif'' or ``non'', is correct with probability at least $1-\delta > 0.5$. Then first, a minor lemma:
\begin{equation}
  \delta \geq \frac{1 - \|\vec{Z} - \vec{U}\|_1}{2} . \label{eqn:unif-lower-l1-lemma}
\end{equation}
Proof of the lemma: Letting $\Pr_A[\text{event}]$ be the probability of ``event'' when the oracle is drawn from our distribution, and analogously for $\Pr_U[\text{event}]$:
\begin{align*}
 &\left| \Pr_U[\text{alg says ``unif''}] - \Pr_A[\text{alg says ``unif''}] \right|  \\
  &= \left| \sum_{s \in [n^m]} \Pr_U[\text{alg says ``unif'' on $s$}]\left(\Pr[s \leftarrow \vec{U}] - \Pr[s \leftarrow \vec{Z}]\right) \right|  \\
  &\leq \sum_{s \in [n^m]} \left|\vec{U}_s - \vec{Z}_s\right|  \\
  &= \|\vec{U} - \vec{Z}\|_1 ;
\end{align*}
on the other hand, the first line is lower-bounded by $|1-\delta ~ - ~ \delta| = 1-2\delta$, which proves the lemma (Inequality \ref{eqn:unif-lower-l1-lemma}).

Now we repeat Paninski's construction, slightly generalized for the $\ell_p$ case. We assume $n$ is even; if not, apply the following construction to the first $n-1$ coordinates. The family of distributions is constructed (and sampled from uniformly) as follows. For each $i=1,3,5,\dots$, flip a fair coin. If heads, let $A_i = \frac{1}{n}(1+\alpha)$ and let $A_{i+1} = \frac{1}{n}(1-\alpha)$. If tails, let $A_i = \frac{1}{n}(1-\alpha)$ and let $A_{i+1} = \frac{1}{n}(1+\alpha)$.

Here $\alpha = \epsilon n^{1/q}$. We need to verify that each $A$ so constructed is a valid probability distribution and that $\|A-U\|_p \geq \epsilon$. Since $n \leq \frac{1}{\epsilon^q}$, we have that $\alpha \leq 1$, so our construction does give a valid probability distribution. And $\|A-U\|_p^p = n\left(\frac{\alpha}{n}\right)^p =  n^{1-p} \epsilon^p n^{p/q} = \epsilon^p$.

Now we just need to upper-bound $\|\vec{U} - \vec{Z}\|_1$, and we will be done. Utilize the inequality of Lemma \ref{lemma:pnorm-ineq}, $\|\vec{U} - \vec{Z}\|_1 \leq \|\vec{U} - \vec{Z}\|_2 \sqrt{n^m}$, and upper-bound this $2$-norm. We have
\begin{align}
 \|\vec{U}-\vec{Z}\|_2^2  \nonumber
  &= \sum_{s \in [n^m]}\left(\vec{Z}_s - \frac{1}{n^m}\right)^2  \nonumber \\
  &= \sum_s \left(\vec{Z}_s^2 - \frac{2}{n^m}\vec{Z}_s + \frac{1}{n^{2m}}\right)  \nonumber \\
  &= \left(\sum_s \vec{Z}_s^2\right) - \frac{1}{n^m} . \label{eqn:dist-unif-lower-1}
\end{align}
Now,
\begin{align*}
 \sum_s \vec{Z}_s^2
  &= \sum_s \sum_{A,A'} \frac{1}{2^n} \Pr[s \mid A] \Pr[s \mid A']
\end{align*}
where $A$ and $A'$ are random variables: They are distributions drawn uniformly from our family, each with probability $\frac{1}{2^{n/2}}$ (since we make $n/2$ binary choices).

Let $s_j$, for $j=1,\dots,m$, be the $j$th sample. Now, rearrange:
\begin{align*}
 \sum_s \vec{Z}_s^2
  &= \sum_{A,A'} \frac{1}{2^n} \sum_s \Pr[s \mid A] \prod_{j=1}^m A'_{s_j}
\end{align*}
View the inner sum as follows: After fixing $A$ and $A'$, we take the expectation, over a draw of a sample $s$ from $A$, of the quantity $\Pr[s \mid A']$, which is expanded into the product. But now, each term $A'_{s_j}$ is independent, since the $m$ samples are drawn i.i.d. from $A$ (and recall that, in this expectation, $A$ and $A'$ are fixed and not random). The expectation of the product is the product of the expectations:
\begin{align*}
 \sum_s \vec{Z}_s^2
 &= \sum_{A,A'} \frac{1}{2^n} \prod_{j=1}^m \sum_s \Pr[s \mid A] A'_{s_j}  \\
 &= \sum_{A,A'} \frac{1}{2^n} \prod_{j=1}^m \sum_{s_j \in [n]} \Pr[s_j \mid A] A'_{s_j}  \\
 &= \sum_{A,A'} \frac{1}{2^n} \prod_{j=1}^m \sum_{i=1}^n A_i A'_i  \\
 &= \sum_{A,A'} \frac{1}{2^n} \left( \sum_{i=1}^n A_i A'_i \right)^m  \\
\end{align*}
We can simplify the inner sum. After factoring out a $\frac{1}{n}$ from each probability, consider the odd coordinates $i=1,3,5,\dots$. Either $A_i \neq A'_i$, in which case $A_iA'_i = \frac{1}{n^2} (1+\alpha)(1-\alpha) = \frac{1}{n^2}(1-\alpha^2) = A_{i+1}A'_{i+1}$, or $A_i = A'_i$. In this case, $A_iA_i' + A_{i+1}A'_{i+1} = \frac{1}{n^2}\left((1+\alpha)^2 + (1-\alpha)^2\right) = \frac{2}{n^2}(1+\alpha^2)$. So the inner sum is equal to
\begin{align*}
 \sum_{i=1}^n A_i A'_i
  &= \frac{1}{n} \left(1 + \frac{2\alpha^2}{n} \sum_{i=1,3,5,\dots} \sigma_i(A,A') \right).
\end{align*}
where
 \[ \sigma_i(A,A') = \begin{cases} 1 & A_i=A'_i \\ -1 & A_i\neq A'_i \end{cases} . \]
Note that unless $A=A'$, $\sigma_i(A,A')$ has a $0.5$ probability of taking each value, independently for all $i$.

OK, we now plug the inner sum back in and use the inequality $1+x \leq e^x$:
\begin{align*}
 \sum_s \vec{Z}_s^2
  &= \sum_{A,A'} \frac{1}{2^n} \left( \frac{1}{n} \left(1 + \frac{2\alpha^2}{n} \sum_{i=1,3,\dots} \sigma_i(A,A')\right) \right)^m  \\
  &\leq \frac{1}{n^m} \sum_{A,A'} \frac{1}{2^n} e^{\frac{2m\alpha^2}{n} \sum_{i=1,3,\dots} \sigma_i(A,A')}  \\
  &= \frac{1}{n^m} \sum_{A,A'} \frac{1}{2^n} \prod_{i=1,3,\dots} e^{\frac{2m\alpha^2}{n} \sigma_i(A,A')} .
\end{align*}
This double sum is an expectation over the random variables $A$ and $A'$, which now means it is an expectation only over the $\sigma_i(A,A')$s. As each is independent and uniform on $\{-1,1\}$, we can convert the expectation of products into a product of expectations, take the expectation, and use the cosh inequality $\frac{e^x + e^{-x}}{2} \leq e^{x^2/2}$:
\begin{align*}
 \sum_s \vec{Z}_s^2
  &\leq \frac{1}{n^m} \prod_{i=1,3,\dots} \E e^{\frac{2m\alpha^2}{n} \sigma_i(A,A')}  \\
  &=    \frac{1}{n^m} \left( \frac{1}{2} e^{\frac{2m\alpha^2}{n}} + \frac{1}{2} e^{\frac{-2m\alpha^2}{n}}\right)^{n/2}  \\
  &\leq \frac{1}{n^m} \left( e^{\frac{2m^2\alpha^4}{n^2}} \right)^{n/2}  \\
  &=    \frac{1}{n^m} e^{\frac{m^2\alpha^4}{n}} .
\end{align*}
Plugging this all the way back into Equation \ref{eqn:dist-unif-lower-1},
\begin{align*}
 \|\vec{U}-\vec{Z}\|_2^2        &\leq \frac{1}{n^m}\left(e^{\frac{m^2\alpha^4}{n}} - 1\right)  \\
 \implies \|\vec{U}-\vec{Z}\|_1 &\leq \frac{1}{\sqrt{n^m}} \sqrt{e^{\frac{m^2\alpha^4}{n}}-1} \sqrt{n^m}  \\
                                &=    \sqrt{e^{\frac{m^2\alpha^4}{n}}-1} .
\end{align*}
It is already apparent that we need $m \geq \Omega\left(\frac{\sqrt{n}}{\alpha^2}\right)$, and by construction $\frac{\sqrt{n}}{\alpha^2} = \frac{\sqrt{n}}{(n^{1/q} \epsilon)^2}$. More precisely, plugging in to Inequality \ref{eqn:unif-lower-l1-lemma} (the ``mini-lemma''), we find that to succeed with probability $\geq 1-\delta$, an algorithm must draw
 \[ m \geq \sqrt{\ln\left((1-2\delta)^2+1\right)} \frac{\sqrt{n}}{\left(n^{1/q}\epsilon\right)^2}  \]
samples.
\end{proof}

\section{Uniformity Testing for \protect\uniftestptwohdr}

\subsection{Lower Bounds (necessary)}

\begin{theorem} \label{theorem:test-unif-lower-infty-eps}
To test uniformity in $\ell_{\infty}$ distance for any $n > \frac{1}{\epsilon}$ requires the following number of samples:
  \[ m = \frac{1-2\delta}{2} \frac{1}{\epsilon}  . \]
\end{theorem}
\textit{Proof sketch.} The proof is similar to the proof of Theorem \ref{theorem:test-unif-lower-smallp-bign}, the lower bound for $p \leq 2$ and $n \geq \frac{1}{\epsilon^q}$. In this case, we only need one distribution $A$ (not a family of distributions), which has probability $\frac{1}{n}+\epsilon$ on one coordinate and is uniform on the others. Thus, $\|A-U_n\|_{\infty} = \epsilon$. Without enough samples, probably the large coordinate is never drawn; but conditioned on this, $A$ and $U_n$ are indistinguishable.
\begin{proof}
Let
 \[ A = \left(\frac{1}{n} + \epsilon, ~ \frac{1}{n} - \frac{\epsilon}{n-1}, ~ \dots, ~ \frac{1}{n} - \frac{\epsilon}{n-1}\right) . \]
If $m \leq \frac{1-2\delta}{2}\frac{1}{\epsilon}$, then
\begin{align*}
 \Pr_A[\text{sample coord 1}]
   &= m\left(\frac{1}{n} + \epsilon\right)  \\
   &< 2m\epsilon  \\
   &\leq 1-2\delta
\end{align*}
using that $\frac{1}{n} < \epsilon$. Also note that
 \[  \Pr_U[\text{sample coord 1}] \leq \Pr_A[\text{sample coord 1}] \leq 1-2\delta . \]

Now, we claim that, conditioned on not sampling coordinate 1, the distribution of samples is the same under $A$ and under $U$. This follows because, for both $A$ and $U$, the distribution over samples conditioned on not sampling coordinate $1$ is uniform.
%Thus, if $X_1$ is the number of samples of coordinate $1$ and $s$ is a generic set of samples, we have
%\begin{align*}
% \Pr_A[\text{say ``uniform''} \mid X_1 = 0]
%  &= \sum_s \Pr_A[\text{sample is }s \mid X_1 = 0] \Pr[\text{say ``uniform'' on }s]  \\
%  &= \sum_s \Pr_U[\text{sample is }s \mid X_1 = 0] \Pr[\text{say ``uniform'' on }s]  \\
% &= \Pr_U[\text{say ``uniform''} \mid X_1 = 0 .
%\end{align*}
Let $\gamma$ be the probability that the algorithm says ``uniform'' given that the samples do not contain coordinate $1$ (again, we just argued that this probability is equal to $\gamma$ whether the distribution is $A$ or $U$). If $\gamma \geq \frac{1}{2}$, then the probability of correctness when drawing samples from $A$ is at most
\begin{align*}
 &\Pr_A[\text{sample coord 1}] + (1-\gamma) \left(1-\Pr_A[\text{sample coord 1}]\right)  \\
 &\leq \frac{1}{2} + \Pr_A[\text{sample coord 1}]\left(1 - \frac{1}{2}\right)  \\
 &< \frac{1}{2}\left(1 + 1-2\delta\right)  \\
 &= 1 - \delta . 
\end{align*}
Similarly, if $\gamma \leq \frac{1}{2}$, then the probability of correctness when drawing samples from $U$ is at most
\begin{align*}
 &\Pr_U[\text{sample coord 1}] + \gamma \left(1-\Pr_U[\text{sample coord 1}]\right)  \\
 &< 1 - \delta
\end{align*}
by the same arithmetic. So the algorithm has a larger failure probability than $\delta$ in at least one of these cases.
\end{proof}

\begin{theorem} \label{theorem:test-unif-lower-infty-n}
To test uniformity in $\ell_{\infty}$ distance for any $n$ requires at least the following number of samples:
 \[ m = \frac{1}{2} \frac{\ln\left(1+n(1-2\delta)^2\right)}{\epsilon^2 n} . \]
\end{theorem}
\begin{proof}
We proceed by the same general technique as in Theorem \ref{theorem:test-unif-lower-smallp-smalln}, the proof of Paninski in \cite{paninski2008coincidence}.

Our family of distributions will be the possible permutations of the distribution $A$ from the proof of Theorem \ref{theorem:test-unif-lower-infty-eps}; namely, we will have a family of $n$ distributions, each of which puts probability $\frac{1}{n} + \epsilon$ on one coordinate and puts probability $\frac{1}{n} - \frac{\epsilon}{n-1}$ on the remaining coordinates. We select a coordinate $i \in \{1,\dots,n\}$ uniformly at random, which chooses the distribution that puts higher probability on $i$.

As shown in the proof of Theorem \ref{theorem:test-unif-lower-infty-eps}, letting $\vec{Z}$ be the distribution of samples obtained by picking a member of the family and then drawing $m$ samples, and letting $\vec{U}$ be the distribution of samples obtained by drawing $m$ samples from $U$, we have for any algorithm
\begin{align}
  \delta \geq \frac{1 - \|\vec{Z} - \vec{U}\|_1}{2} . \label{eqn:test-l1-mini-lemma2}
\end{align}
Meanwhile, by the $p$-norm inequality (Lemma \ref{lemma:pnorm-ineq}), recalling that $\vec{Z}$ and $\vec{U}$ are vectors of length $n^m$,
\begin{align}
 \|\vec{Z} - \vec{U}\|_1
  &\leq \sqrt{n^m \|\vec{Z}-\vec{U}\|_2^2}  \nonumber \\
  &= \sqrt{n^m \|\vec{Z}\|_2^2 - 1} ,  \label{eqn:test-infty-bound-l1-norm}
\end{align}
using that
\begin{align*}
 \|\vec{Z} - \vec{U}\|_2^2
  &= \sum_s |\vec{Z}_s - \vec{U}_s|^2  \\
  &= \sum_s \vec{Z}_s^2 + \vec{U}_s^2 - 2\vec{Z}_s\vec{U}_s  \\
  &= \|\vec{Z}\|_2^2 + \frac{1}{n^m} - 2\frac{1}{n^m}\sum_s \vec{Z}_s  \\
  &= \|\vec{Z}\|_2^2 - \frac{1}{n^m} .
\end{align*}
Thus, our task is again to bound $\|\vec{Z}\|_2^2$. Our next step toward this will be to obtain the following:
\begin{align*}
  \sum_s \vec{Z}_s^2 &= \E_{A,A'} \left( \E_{s \sim A} \Pr[s \sim A'] \right)^m.
\end{align*}
Here, $A$ and $A'$ are two distributions draw randomly from the family, and the notation $s \sim A$ means drawing a set of samples $s$ i.i.d. from $A$ (so the inner expectation is over a sample $s$ drawn from $A$ and is the expectation of the probability of that sample according to $A'$). The proof is precisely as in that of Theorem \ref{theorem:test-unif-lower-infty-eps}:
\begin{align*}
 \sum_s \vec{Z}_s^2
  &= \sum_s \left(\E_A \Pr[s \sim A]\right)\left(\E_{A'} \Pr[s \sim A']\right)  \\
  &= \E_{A,A'} \sum_s \Pr[s \sim A] \Pr[s \sim A']  \\
  &= \E_{A,A'} \E_{s \sim A} \Pr[s \sim A']  \\
  &= \E_{A,A'} \E_{s \sim A} \prod_{k=1}^m \Pr[s_k \sim A']  \\
  &= \E_{A,A'} \left( \E_{s \sim A} \Pr[s_k \sim A'] \right)^m.
\end{align*}
We used that each sample $s_k$ in $s$ is independent, so the expectation of the product is the product of the expectations; and since they are identically distributed, this is just the inner expectation to the $m$th power.

Next, we claim that
 \[ \E_{s\sim A} \Pr[s_k \sim A'] = \begin{cases} \frac{1}{n} + \frac{\epsilon^2 n}{n-1}  & A = A'  \\
      \frac{1}{n} - \frac{\epsilon^2 n}{(n-1)^2}  & A \neq A' \end{cases} .  \]
To prove it, suppose that $A$ has highest probability on coordinate $i$ and $A'$ on coordinate $j$. Then
\begin{align*}
  &\E_{s\sim A} \Pr[s_k \sim A']  \\
  &= \Pr[j \sim A] \left(\frac{1}{n} + \epsilon\right) + \left(1 - \Pr[j \sim A]\right)\left(\frac{1}{n} - \frac{\epsilon}{n-1}\right)
\end{align*}
and since $\Pr[j \sim A]$ is either $\frac{1}{n} + \epsilon$ in the case $A=A'$ or else $\frac{1}{n} - \frac{\epsilon}{n-1}$ otherwise, one can check the claim.

Thus we now have
\begin{align*}
 \sum_s \vec{Z}_s^2
  &= \E_{A,A'} \left( \begin{cases} \frac{1}{n} + \frac{\epsilon^2 n}{n-1}  & A = A'  \\
      \frac{1}{n} - \frac{\epsilon^2 n}{(n-1)^2}  & A \neq A' \end{cases} \right)^m .
\end{align*}
And because $A=A'$ with probability exactly $\frac{1}{n}$ when both are chosen randomly,
\begin{align*}
 \sum_s \vec{Z}_s^2
  &= \frac{1}{n}\left(\frac{1}{n} + \frac{\epsilon^2 n}{n-1}\right)^m + \frac{n-1}{n}\left(\frac{1}{n} - \frac{\epsilon^2 n}{(n-1)^2}\right)^m  \\
  &= \frac{1}{n^m} \left( \frac{1}{n} \left(1 + \frac{\epsilon^2 n^2}{n-1}\right)^m + \frac{n-1}{n}\left(1 - \frac{\epsilon^2 n^2}{(n-1)^2}\right)^m \right) \\
  &\leq \frac{1}{n^m} \left( \frac{1}{n} \left(1 + 2\epsilon^2 n\right)^m + \frac{n-1}{n} \right) \\
  &\leq \frac{1}{n^m} \left( \frac{1}{n}\exp\left[ 2m\epsilon^2 n \right]  + \frac{n-1}{n} \right)  \\
  &=    \frac{1}{n^m} \left( \frac{1}{n}\left(\exp\left[ 2m\epsilon^2 n \right] - 1\right) + 1 \right) .
\end{align*}

Plugging back in to Inequalities \ref{eqn:test-infty-bound-l1-norm} and \ref{eqn:test-l1-mini-lemma2}, it is necessary that
 \[ \delta \geq \frac{1 - \sqrt{\frac{1}{n}\left(\exp\left[ 2m\epsilon^2 n \right] - 1\right)}}{2} ; \]
equivalently,
 \[ \frac{1}{n}\left(\exp\left[ 2m\epsilon^2 n \right] - 1\right) \geq (1 - 2\delta)^2 ; \]
which equates to
 \[ \exp\left[ 2m\epsilon^2 n \right] \geq n(1 - 2\delta)^2 + 1 . \]
Thus,
 \[ m \geq \frac{1}{2} \frac{\ln\left(1 + n(1-2\delta)^2 \right)}{\epsilon^2 n} . \]
\end{proof}

\subsection{Upper Bounds (sufficient)}

Let us briefly recall Algorithm \ref{alg:test-unif-infty}. For a threshold $\alpha(n) = \bigTheta{\frac{\ln(n)}{n}}$, we condition on whether $\epsilon \leq 2\alpha(n)$ or $\epsilon > \alpha(n)$. These essentially correspond to the small $n$ and large $n$ regimes for this problem.

If $\epsilon \leq 2\alpha(n)$, we draw $\bigTheta{\frac{\ln(n)}{n \epsilon^2}}$ samples and check whether all coordinates have a number of samples close to their expectation; if not, we output ``not uniform''.

If $\epsilon > 2\alpha(n)$, we draw $\bigTheta{\frac{1}{\epsilon}}$ samples. We choose $\hat{n}$ such that $\epsilon = 2\alpha(\hat{n})$; in other words, $\epsilon = \bigTheta{\frac{\ln(\hat{n})}{\hat{n}}}$. We then divide the coordinates into about $\hat{n}$ ``groups'' where, if $A=U$, then each group has probability about $\frac{1}{\hat{n}}$. We then check for any group with a ``large'' outlier number of samples; if one exists, then we output ``not uniform''.

\theoremtestinftysufficient*
\begin{proof}
For each case, we will prove two lemmas that imply the upper bound. First, for the case, $\epsilon \leq 2\alpha(n)$, Lemma \ref{lemma:test-infty-small-equal} states that if $A=U$ then $X_i \in \frac{m}{n} \pm t$ for all coordinates $i$ except with probability $\delta$; and Lemma \ref{lemma:test-infty-small-different} states that if $\|A-U\|_{\infty} \geq \epsilon$ then some coordinate has $X_i \not\in \frac{m}{n} \pm t$ except with probability $\delta$.

Similarly, for the case $\epsilon > 2\alpha$, Lemma \ref{lemma:test-infty-large-equal} states that if $A=U$ then $X_j < m\epsilon - t$ for all groups $j$ except with probability $\delta$; and Lemma \ref{lemma:test-infty-large-different} states that if $\|A-U\|_{\infty} \geq \epsilon$ then some group has $X_j \geq m\epsilon - t$ except with probability $\delta$.
\end{proof}
\begin{lemma} \label{lemma:test-infty-small-equal}
If $A = U$, then (for any $m, n,\epsilon$) with probability at least $1-\delta$, every coordinate $i$ satisfies that $X_i \in \frac{m}{n} \pm \sqrt{3\frac{m}{n}\ln\left(\frac{2n}{\delta}\right)}$.
\end{lemma}
\begin{proof}
The number of samples of any particular coordinate $i$ is distributed as a Binomial$(m,1/n)$. Let $\mu = \E X_i = \frac{m}{n}$. By a Chernoff bound (e.g. Mitzenmacher and Upfal~\cite{mitzenmacher2005probability}, Theorems 4.4 and 4.5), the following inquality holds for both $P = \Pr[X_i \leq \mu - t]$ and $P = \Pr[X_i \geq \mu + t]$:
\begin{equation}
  P \leq e^{-\frac{t^2}{3\mu}} . \label{eqn:chernoff-loose}
\end{equation}
Since $\mu = \frac{m}{n}$, if we set
 \[ t = \sqrt{3\frac{m}{n}\ln\left(\frac{2n}{\delta}\right)} , \]
then we get that $X_i$ falls outside the range in either direction with probability at most $\frac{\delta}{n}$; a union bound over the $n$ coordinates gives that the probability of any of them falling outside the range is at most $\delta$.
\end{proof}

\begin{lemma} \label{lemma:test-infty-small-different}
Suppose $\|A-U\|_{\infty} \geq \epsilon$ and $\epsilon \leq 2\alpha(n)$, and we draw $m \geq 23\frac{\ln\left(\frac{2n}{\delta}\right)}{n \epsilon^2}$ samples. Then with probability at least $1-\delta$, some coordinate $i$ satisfies that $X_i \not\in \frac{m}{n} \pm \sqrt{3\frac{m}{n}\ln\left(\frac{2n}{\delta}\right)}$.
\end{lemma}

\begin{proof}
There must be some coordinate $i$ such that either $A_i \leq \frac{1}{n} - \epsilon$ or $A_i \geq \frac{1}{n} + \epsilon$. Take the first case. (Note that in this case $\frac{1}{n} \geq \epsilon$.)
By the Chernoff bound mentioned above (Inequality \ref{eqn:chernoff-loose}),
\begin{align*}
 \Pr\left[ X_i \geq \frac{m}{n} - t \right]
  &= \Pr\left[ X_i \geq \E X_i + \left(\frac{m}{n} - t - \E X_i\right) \right] \\
  &\leq \exp \left[ -\frac{\left(\frac{m}{n} - t - \E X_i\right)^2 }{ 3\E X_i} \right]  \\
  &\leq \exp \left[ -\frac{\left(m\epsilon - t\right)^2 }{ 3m\left(\frac{1}{n} - \epsilon\right)} \right]
\end{align*}
because $\E X_i \leq m\left(\frac{1}{n} - \epsilon\right)$ and this substitution only increases the bound.

For this to be bounded by $\delta$, it suffices that
 \[ m\epsilon - t \geq \sqrt{3 \frac{m}{n} \ln\left(\frac{1}{\delta}\right)} . \] 
Now we substitute $t = \sqrt{3\frac{m}{n}\ln\left(\frac{2n}{\delta}\right)}$. Because $t$ is larger than the right-hand side, it suffices that
\begin{align*}
 m\epsilon &\geq 2 t  \\
 \iff m\epsilon &\geq 2\sqrt{\frac{3m}{n}\ln\left(\frac{2n}{\delta}\right)}  \\
 \iff m &\geq \frac{12 \ln\left(\frac{2n}{\delta}\right)}{n\epsilon^2} .
\end{align*}
That completes the proof for this case.

Now take the case that there exists some $A_i \geq \frac{1}{n} + \epsilon$.
\begin{align*}
 \Pr\left[ X_i \leq \frac{m}{n} + t \right]
  &= \Pr\left[ X_i \leq \E X_i - \left(\E X_i - \frac{m}{n} - t\right) \right]  \\
  &\leq \exp\left[-\frac{(\E X_i - \frac{m}{n} - t)^2}{3 \E X_i} \right] .
\end{align*}
This bound is decreasing in $\E X_i$, so we can use the inequality $\E X_i \geq m\left(\epsilon + 1/n\right)$:
\begin{align*}
 &\leq \exp\left[-\frac{(m\epsilon - t)^2}{3m\left(\frac{1}{n} + \epsilon\right)}\right] .
\end{align*}
The above is bounded by $\delta$ if it is true that
 \[ m\epsilon - t \geq \sqrt{3m\ln\left(\frac{1}{\delta}\right)\left(\frac{1}{n} + \epsilon\right)} . \]
Since $\epsilon \leq 2\alpha(n)$, we have
\begin{align*}
 \ln\left(\frac{1}{\delta}\right)\left(\frac{1}{n} + \epsilon\right)
  &\leq \ln\left(\frac{1}{\delta}\right)\left(\frac{1}{n} + \myalpha{n}\right)  \\
  &= \frac{3\ln\left(\frac{1}{\delta}\right)}{n} + \frac{2\ln\left(2n\right)}{n}  \\
  &\leq \frac{3\ln\left(\frac{2n}{\delta}\right)}{n}  .
\end{align*}
Thus, it suffices to have $m$ satisfy
\begin{align*}
 m\epsilon - t
  &\geq \sqrt{\frac{9m\ln\left(\frac{2n}{\delta}\right)}{n}}  \\
  &=    3\sqrt{\frac{m\ln\left(\frac{2n}{\delta}\right)}{n}} .
\end{align*}
Because $t = \sqrt{3}\sqrt{\frac{m\ln\left(\frac{2n}{\delta}\right)}{n}}$, it suffices that
\begin{align*}
 m\epsilon \geq \left(3+\sqrt{3}\right)\sqrt{\frac{m\ln\left(\frac{2n}{\delta}\right)}{n}}  \\
 \iff m \geq \left(3+\sqrt{3}\right)^2 \frac{\ln\left(\frac{2n}{\delta}\right)}{n \epsilon^2} .
\end{align*}
In particular, $\left(3 + \sqrt{3}\right)^2 \leq 23$.
\end{proof}

\begin{lemma} \label{lemma:test-infty-large-equal}
Suppose $A = U$ and $\epsilon > 2\alpha(n)$, and we draw $m \geq 35\frac{\ln\left(1/\delta\right)}{\epsilon}$ samples. Then with probability at least $1-\delta$, every group $j$ satisfies that $X_{j} \leq m\epsilon - \sqrt{3m\epsilon\ln\left(\frac{1}{\delta}\right)}$.
\end{lemma}

\begin{proof}
Recall that we have divided into at most $2\hat{n}$ groups, each of size $\lfloor \frac{n}{\hat{n}} \rfloor$. When $A=U$, this implies that each group has probability at most $\frac{1}{\hat{n}}$. Therefore, by the same Chernoff bound (Inequality \ref{eqn:chernoff-loose}), for any group $j$,
\begin{align*}
 \Pr[X_j \geq m\epsilon - t] 
  &= \Pr\left[X_j \geq \E X_j + \left(m\epsilon - t - \E X_j \right) \right]  \\
  &\leq \exp\left[ -\frac{\left(m\epsilon - t - \E X_j \right)^2}{3\E X_j} \right]  \\
  &\leq \exp\left[ -\frac{\left(m\epsilon - t - \frac{m}{\hat{n}}\right)^2}{3m/\hat{n}} \right] .
\end{align*}
We wish this probability to be bounded by $\frac{\delta}{2\hat{n}}$, as then, by a union bound over the at most $2\hat{n}$ groups, the probability that any group exceeds the threshold is at most $\delta$. Thus, it suffices that
\begin{align*}
 m\epsilon - t - \frac{m}{\hat{n}}
  &\geq \sqrt{3\frac{m}{\hat{n}}\ln\left(\frac{2\hat{n}}{\delta}\right)}
\end{align*}
Now we can apply our fortuitous choice of $\hat{n}$: Note that
\begin{align*}
 \frac{\ln\left(\frac{2\hat{n}}{\delta}\right)}{\hat{n}}
  &= \frac{\ln\left(\frac{1}{\delta}\right) + \ln(2\hat{n})}{\hat{n}}  \\
  &= \ln\left(\frac{1}{\delta}\right) \alpha(\hat{n})  \\
  &= \ln\left(\frac{1}{\delta}\right) \frac{\epsilon}{2} .
\end{align*}
So it suffices that
\begin{align*}
 m\epsilon - t - \frac{m}{\hat{n}}
  &\geq \sqrt{\frac{3}{2}}\sqrt{m\epsilon \ln\left(\frac{1}{\delta}\right)} .
\end{align*}
We have that $t = \sqrt{3}\sqrt{m\epsilon \ln\left(\frac{1}{\delta}\right)}$, so it suffices that
 \[ m\left(\epsilon - \frac{1}{\hat{n}}\right) \geq \sqrt{3}\left(1+\frac{1}{\sqrt{2}}\right)\sqrt{m\epsilon \ln\left(\frac{1}{\delta}\right)} . \]
Since $\epsilon = 2\alpha(\hat{n})$, in particular $\epsilon \geq \frac{2}{\hat{n}}$, or $\epsilon - \frac{1}{\hat{n}} \geq \frac{\epsilon}{2}$. Therefore, it suffices that
\begin{align*}
  m\epsilon &\geq 2\sqrt{3}\left(1+\frac{1}{\sqrt{2}}\right)\sqrt{m\epsilon \ln\left(\frac{1}{\delta}\right)}  \\
 \iff ~~ m &\geq \left(2\sqrt{3}\left(1+\frac{1}{\sqrt{2}}\right)\right)^2 \frac{\ln\left(\frac{1}{\delta}\right)}{\epsilon} .
\end{align*}
In particular, $\left(2\sqrt{3}\left(1+\frac{1}{\sqrt{2}}\right)\right)^2 \leq 35$.
\end{proof}

\begin{lemma} \label{lemma:test-infty-large-different}
Suppose $\|A-U\|_{\infty} \geq \epsilon$ and $\epsilon > 2\alpha(n)$. Then (for any $m$) with probability at least $1-\delta$, there exists some group $j$ whose number of samples $X_j \geq m\epsilon - \sqrt{3m\epsilon\ln\left(\frac{1}{\delta}\right)}$.
\end{lemma}

\begin{proof}
This is just a Chernoff bound. Note that if coordinate $i$ has some number of samples, then there exists a group (that containing $i$) having at least that many samples. So we simply prove the lemma for the number of samples of some coordinate $X_i$.

If $\|A-U\|_{\infty} \geq \epsilon$ and $\epsilon > 2\alpha(n)$, then in particular $\epsilon > \frac{2}{n}$, which implies that there exists some coordinate $i$ with $A_i > \frac{1}{n} + \epsilon$ (because $\frac{1}{n} - \epsilon < 0$).
Using the Chernoff bound mentioned above (Inequality \ref{eqn:chernoff-loose}),
\begin{align*}
 \Pr[X_i < m\epsilon - t]
  &=    \Pr[X_i < \E X_i - (\E X_i - m\epsilon + t)]  \\
  &\leq \exp\left[ - \frac{\left(\E X_i - m\epsilon + t\right)^2}{3\E X_i} \right]  \\
  &\leq \exp\left[ - \frac{t^2}{3m\epsilon} \right] ,
\end{align*}
using that $\E X_i \geq m\epsilon$; and this is bounded by $\delta$ if
 \[ t \geq \sqrt{3m\epsilon\ln\left(\frac{1}{\delta}\right)} . \]
\end{proof}

\section{Distribution Learning}
\label{section:learning}

\subsection{Upper Bounds (sufficient)}

We first show the following bound for $\ell_2$ learning, which is slightly tighter than Theorem \ref{theorem:learn-upper-non}.
\begin{theorem} \label{theorem:learn-l2-upper}
To learn in $\ell_2$ distance with failure probability $\delta$, it suffices to run Algorithm \ref{alg:learner} while drawing the following number of samples:
 \[ m = \frac{1}{\delta} \frac{1}{\epsilon^2} . \]
\end{theorem}
Before proving it, let us separately show the key fact:
\begin{lemma} \label{lemma:learn-upper-l2-mean}
If we draw $m$ samples, then
 \[ \E\left[ \|A - \hat{A}\|_2^2 \right] \leq \frac{1}{m} . \]
\end{lemma}
\begin{proof}[of Lemma \ref{lemma:learn-upper-l2-mean}]
As in the proof of Theorem \ref{theorem:learn-upper-non}, letting $X_i$ be the number of samples of coordinate $i$:
\begin{align*}
 \E \|\hat{A}-A\|_2^2
  &= \frac{1}{m^2} \sum_{i=1}^n \E \left(X_i - \E X_i\right)^2  \\
  &= \frac{1}{m^2} \sum_{i=1}^n Var(X_i)  \\
  &= \frac{1}{m^2} \sum_{i=1}^n m A_i (1-A_i)  \\
  &\leq \frac{1}{m} \sum_{i=1}^n A_i  \\
  &= \frac{1}{m} .
\end{align*}
\end{proof}
\begin{proof}[of Theorem \ref{theorem:learn-l2-upper}]
Using Markov's Inequality and Lemma \ref{lemma:learn-upper-l2-mean},
\begin{align*}
 \Pr[\|\hat{A}-A\|_2 \geq \epsilon]
  &= \Pr[\|\hat{A}-A\|_2^2 \geq \epsilon^2]  \\
  &\leq \frac{\E\|\hat{A}-A\|_2^2}{\epsilon^2}  \\
  &\leq \frac{1}{m \epsilon^2}  \\
  &= \delta
\end{align*}
if $m = \frac{1}{\delta}\frac{1}{\epsilon^2}$.
\end{proof}

\theoremlearnupperusetwo*
\begin{proof}
For the case $p \geq 2$, we have (Lemma \ref{lemma:pnorm-ineq}) that $\|A-\hat{A}\|_p \leq \|A - \hat{A}\|_2$, so learning to within $\epsilon$ in $\ell_2$ distance implies learning for $\ell_p$ distance.

For $p \leq 2$: By Theorem \ref{theorem:learn-l2-upper}, if we run Algorithm \ref{alg:learner} while drawing $\frac{1}{\delta}\frac{1}{\alpha^2}$ samples, then with probability $1-\delta$, $\|\hat{A}-A\|_2 \leq \alpha$.

In this case, by the $\ell_p$ norm inequality of Lemma \ref{lemma:pnorm-ineq}, for $p \leq 2$,
\begin{align*}
 \|\hat{A}-A\|_p
  &\leq n^{\frac{1}{p}-\frac{1}{2}} \|\hat{A}-A\|_2  \\
  &=    \frac{\sqrt{n}}{n^{1/q}} \|\hat{A}-A\|_2  \\
  &\leq \frac{\sqrt{n}}{n^{1/q}} \alpha  \\
  &=    \epsilon
\end{align*}
if we set $\alpha = \frac{\epsilon n^{1/q}}{\sqrt{n}}$. Thus, we are guaranteed correctness with probability $1-\delta$ if we draw a number of samples equal to
 \[ \frac{1}{\delta} \frac{1}{\alpha^2} = \frac{1}{\delta} \frac{n}{\left(n^{1/q}\epsilon\right)^2} . \]
This says that the above number of samples is sufficient. However, in the large $n$ regime, we can do better: By the $\ell_p$ norm inequality of Lemma \ref{lemma:pnorm-leq-2}, using that $\|\hat{A}-A\|_1 \leq 2$,
\begin{align*}
 \|\hat{A}-A\|_p^q
  &\leq 2^{q-2} \|\hat{A} - A\|_2^2  \\
  &\leq \frac{2^q}{4} \alpha^2  \\
  &\leq \epsilon^q
\end{align*}
if we set $\alpha^2 = 4\frac{\epsilon^q}{2^q}$; but then we are guaranteed correctness with probability $1-\delta$ if we draw
 \[ m = \frac{1}{\delta}\frac{1}{\alpha^2} = \frac{1}{\delta} \frac{1}{4} \left(\frac{2}{\epsilon}\right)^q  \]
samples. This number of samples is also unconditionally sufficient; we find that the first is better (smaller) bound when $n \leq \left(\frac{2}{\epsilon}\right)^q$.
\end{proof}

A logarithmic dependence on the failure probability $\delta$ is possible, in two steps.\footnote{This idea is also folklore and not original to this paper.} First, if we draw enough samples, then the expected $\ell_p$ distance between $A$ and $\hat{A}$ (the empirical distribution) is less than $\epsilon/2$. Second, if we draw enough samples, then this $\ell_p$ distance is concentrated within $\epsilon/2$ of its expectation.
These two steps are formalized in the next two lemmas.
\begin{lemma} \label{lemma:learn-upper-lp-mean}
For $1 \leq p \leq 2$, if we draw $m$ samples, then
 \[ \E \|\hat{A}-A\|_p \leq \min\left\{ \sqrt{\frac{n}{n^{2/q} m}} ~,~ \frac{2}{2^{2/q} m^{1/q}} \right\} . \]
\end{lemma}
\begin{proof}
Lemma \ref{lemma:learn-upper-l2-mean} stated that $\E \|\hat{A}-A\|_2^2 \leq \frac{1}{m}$. By Jensen's inequality, $\left(\E \|\hat{A}-A\|_2\right)^2 \leq \E \|\hat{A}-A\|_2^2$, so $\E\|\hat{A}-A\|_2 \leq \sqrt{\frac{1}{m}}$. By Lemma \ref{lemma:pnorm-ineq} (the $\ell_p$-norm inequality), this implies
 \[ \E \|\hat{A}-A\|_p \leq \sqrt{\frac{1}{m}} n^{\frac{1}{p}-\frac{1}{2}} , \]
which by some rearranging (using $\frac{1}{p} = 1 - \frac{1}{q}$) gives half of the lemma.
Now by Lemma \ref{lemma:pnorm-leq-2},
\begin{align*}
 \E \|\hat{A}-A\|_p^q
  &\leq 2^{q-2} \E \|\hat{A}-A\|_2^2  \\
  &\leq 2^{q-2} \frac{1}{m} ,
\end{align*}
implying by Jensen's inequality that $\E \|\hat{A}-A\|_p \leq 2^{\frac{q-2}{q}} / m^{1/q}$.
Rearranging gives the lemma.
\end{proof}
\begin{corollary} \label{corollary:learn-upper-lp-mean}
We have $\E \|\hat{A}-A\|_p \leq \frac{\epsilon}{2}$ if
 \[ m \geq \min\left\{\frac{4n}{\left(n^{1/q}\epsilon\right)^2} ~,~ \frac{1}{4}\left(\frac{4}{\epsilon}\right)^q \right\} . \]
\end{corollary}

\begin{lemma} \label{lemma:learn-upper-lp-concentration}
If we draw $m$ samples, then
 \[ \Pr\left[ \|\hat{A}-A\|_p \geq \E \|\hat{A}-A\|_p + \frac{\epsilon}{2} \right] \leq e^{-m\epsilon^2 / 2^{\frac{2}{p} + 1}} . \]
\end{lemma}
\begin{proof}
We will simply apply McDiarmid's inequality.
Letting $Y_i$ denote the $i$th sample, we can let $f(Y_1,\dots,Y_m) = \|A - \hat{A}\|_p$.
McDiarmid's\footnote{An application of the Azuma-Hoeffding martingale inequality method, \emph{e.g.} Mitzenmacher and Upfal~\cite{mitzenmacher2005probability}, Section 12.5.} states that, if changing any $Y_i$ changes the value of $f$ by at most $c$, then
 \[ \Pr[ f(Y_1,\dots,Y_m) \geq \E f(Y_1,\dots,Y_m) + t] \leq \exp\left[\frac{-2t^2}{m c^2} \right] . \]

In our case, changing any $Y_i$ changes the value of $f$ by at most $\frac{2^{1/p}}{m}$, argued as follows. Let $D \in \mathbb{R}^n$ be a vector with two nonzero entries, one of them $\frac{1}{m}$ and the other $\frac{-1}{m}$. Changing one sample $Y_i$ changes the empirical distribution to $\hat{A} + D$ for some such $D$, so the new value of $f$ is $\|A - (\hat{A} + D)\|_p \in \|A - \hat{A}\|_p \pm \|D\|_p$ by the triangle inequality, and $\|D\|_p = \frac{2^{1/p}}{m}$.

McDiarmid's inequality then states that
\begin{align*}
 &\Pr[ f(Y_1,\dots,Y_m) \geq \E f(Y_1,\dots,Y_m) + t]  \\
 &\leq \exp\left[\frac{-2t^2 }{ m \left(2^{1/p}/m\right)^2} \right]  \\
 &= \exp\left[\frac{-mt^2 }{ 2^{\frac{2}{p}-1}} \right] ,
\end{align*}
and we plug in $t = \frac{\epsilon}{2}$.
\end{proof}
\begin{corollary} \label{corollary:learn-upper-lp-concentration}
We have $\Pr\left[ \|\hat{A}-A\|_p \geq \E \|\hat{A}-A\|_p + \frac{\epsilon}{2} \right] \leq \delta$ if
 \[ m \geq \frac{2^{\frac{2}{p}+1}\ln(1/\delta)}{\epsilon^2} . \]
\end{corollary}

\begin{theorem} \label{theorem:learn-upper-l2-logdelta}
For learning in $\ell_p$ distance for $p \geq 2$ with failure probability $\delta \leq \frac{1}{e}$, it suffices to run Algorithm \ref{alg:learner} while drawing the following number of samples:
 \[ m = \frac{4 \ln(1/\delta)}{\epsilon^2} . \]
\end{theorem}
\begin{proof}
First, note that it suffices to prove the theorem for $\ell_2$ distance, because for $p \geq 2$, $\|\hat{A}-A\|_p \leq \|\hat{A}-A\|_2$.
Now, for $\ell_2$ distance, it suffices that $\E\|\hat{A}-A\|_2 \leq \frac{\epsilon}{2}$ and that, with probability $1-\delta$, $\|\hat{A}-A\|_2$ exceeds its expectation by at most $\frac{\epsilon}{2}$. Therefore, Corollaries \ref{corollary:learn-upper-lp-mean} and \ref{corollary:learn-upper-lp-concentration} state that it suffices to have
 \[ m \geq \max\left\{ \frac{4}{\epsilon^2} ~,~ \frac{4\ln(1/\delta)}{\epsilon^2} \right\} . \]
\end{proof}

\theoremlearnupperlogdelta*
\begin{proof}
It suffices that $\E\|\hat{A}-A\|_p \leq \frac{\epsilon}{2}$ and $\|\hat{A}-A\|_p$ exceeds its expectation by at most $\frac{\epsilon}{2}$.
Thus, the bounds follow directly from Corollaries \ref{corollary:learn-upper-lp-mean} and \ref{corollary:learn-upper-lp-concentration}.
\end{proof}

\subsection{Lower Bounds (necessary)} \label{section:learn-lower}

The lower bounds as stated below are proven in this section, but can also be deduced from folklore as follows.

\theoremlearnlower*

\begin{proof}
For $p \geq 2$, we can deduce this bound from the fact that distinguishing a $2\epsilon$-biased coin from uniform requires $\Omega{\frac{1}{\epsilon^2}}$ samples. This reduction is proven formally in Theorem \ref{theorem:learning-lower-from-coin}.

In Theorem \ref{theorem:learning-lower-our-bounds}, we prove the remaining bounds in this theorem. However, bounds at least this good can apparently be deduced from folklore as follows. It is ``known'' that learning in $\ell_1$ distance requires $\bigOmega{\frac{n}{\epsilon^2}}$ samples. If we interpret this statement to hold for every fixed $\delta$ (the author is unsure if this is the correct interpretation), then we get bounds that match the upper bounds up to constant factors for every fixed $p,\delta$: By Lemma \ref{lemma:pnorm-ineq} an $\ell_p$ learner to within distance $\epsilon$ is an $\ell_1$ learner to within distance $\epsilon n^{1/q}$. $\ell_p$ learning therefore requires $\bigOmega{\frac{n}{\left(\epsilon n^{1/q}\right)^2}}$ samples. Now, for $1 < p < \infty$, if $n \geq \frac{1}{\epsilon^q}$, note that learning on support size $n$ is at least as hard as learning on support size $\hat{n} < n$, so by setting $\hat{n}$ to be the optimal $\frac{1}{\epsilon^q}$ in the previous bound, we get the lower bound $\bigOmega{\hat{n}} = \bigOmega{\frac{1}{\epsilon^q}}$.

Regardless of the folklore fact, we prove the stated lower bounds for these cases ($1 \leq p \leq 2$) in Theorem \ref{theorem:learning-lower-our-bounds}.
\end{proof}
In the small $n$ regime, we will only show that the upper bound is tight as $\delta \to 0$. It is a problem in progress to improve the following approach to give a tighter matching bound. \\

Recall that the general approach is to first construct a ``large'' set of distributions $S$, each of pairwise distance at least $2\epsilon$. Then we show a lower bound on the probability of identifying a member of $S$ when it is chosen uniformly and samples are drawn from it.

\begin{lemma} \label{lemma:learning-size-S} For any $p \in [1,\infty]$, for all $\hat{n} \in \mathbb{N}$ and $\epsilon > 0$, there is a set $S$ of probability distributions on $\{1,\dots,\hat{n}\}$ of size at least
 \[ |S| \geq \frac{\Gamma\left(1+\frac{\hat{n}-1}{p}\right)}{(\hat{n}-1)! \left(4\epsilon \Gamma\left(1+\frac{1}{p}\right)\right)^{\hat{n}-1}}  \]
with pairwise $\ell_p$ distance greater than $2\epsilon$, \emph{i.e.} $\|A-B\|_p > 2\epsilon$ for all pairs $A\neq B$ in $S$.
\end{lemma}
\begin{proof} By a sphere packing argument as with, \emph{e.g.}, the Gilbert-Varshamov bound in the field of error-correcting codes.

Each probability distribution is a point in the $\hat{n}$-dimensional simplex, which is the set $\{A \in \mathbb{R}^{\hat{n}} : \sum_i A_i = 1, A_i \geq 0 \forall i\}$.
Now, suppose we have a ``maximal packing'' of distributions that are at least $2\epsilon$ apart; that is, we have a set $S$ of points in this simplex such that:
\begin{enumerate}
 \item For all pairs $A,B \in S$, $\|A - B\|_p > 2\epsilon$, and
 \item Adding any point in the simplex to $S$ violates this condition.
\end{enumerate}
Then for any point $x$ in the simplex, there exists at least one $A \in S$ with $\|A-x\|_p \leq 2\epsilon$. (Otherwise, we could add $x$ to $S$ without violating the condition.) In other words, every point in the simplex is contained in an $\ell_p$ ball of radius $2\epsilon$ around some member of $S$, or
 \[ \text{$n$-dimensional simplex} \subseteq \bigcup_{A \in S} \{y : \|A-y\|_p \leq 2\epsilon\} \]
which implies that
 \[ \text{Vol($n$-dimensional simplex)} \leq |S| \text{Vol($\ell_p$ ball of radius $2\epsilon$)} . \]
The volume of an $\ell_p$ ball of radius $r$ in $k$-dimensional space is $(2r)^k \Gamma\left(1+\frac{1}{p}\right)^k / \Gamma\left(1+\frac{k}{p}\right)$, where the Gamma function $\Gamma$ is the generalization of the factorial function, with $\Gamma(x) = (x-1)!$ for positive integers $x$.

Viewing the $\hat{n}$-dimensional simplex as a set in $\hat{n}-1$-dimensional space, it has volume $\frac{1}{(\hat{n}-1)!}$. Meanwhile, the $\ell_p$ balls in the simplex also lie in $\hat{n}-1$-dimensional space. So we obtain the inequality
\begin{align*}
 |S| &\geq \frac{\text{Vol($\hat{n}$-dimensional simplex)}}{\text{Vol($\ell_p$ ball of radius $2\epsilon$)}}  \\
     &= \frac{1/(\hat{n}-1)!}{(4\epsilon)^{\hat{n}-1} \Gamma\left(1+\frac{1}{p}\right)^{\hat{n}-1}/\Gamma\left(1+\frac{\hat{n}-1}{p}\right)}  \\
     &= \frac{\Gamma\left(1+\frac{\hat{n}-1}{p}\right)}{(\hat{n}-1)! \left(4\epsilon \Gamma\left(1+\frac{1}{p}\right)\right)^{\hat{n}-1}} .
\end{align*}
\end{proof}

\begin{corollary} \label{corollary:learning-size-S} There exists a set $S$ of distributions with pairwise distance greater than $2\epsilon$ of size
 \[ |S| \geq \begin{cases} \frac{1}{5\epsilon}  & \text{any $p$, $\hat{n}=2$}  \\
                           e^{\frac{p}{12}} \frac{1}{\sqrt{p}}\left(\frac{1}{4(\hat{n}-1)^{1/q}\epsilon}\right)^{\hat{n}-1}  & \text{$p < \infty$, any $\hat{n}$}  .  \end{cases}  \]
\end{corollary}

\begin{proof}
Picking $\hat{n}=2$, we have $\Gamma\left(1+\frac{\hat{n}-1}{p}\right) \geq 0.8856\dots$, which is the minimum of the Gamma function; and $\Gamma\left(1+\frac{1}{p}\right) \leq 1$ for $p \in [1,\infty]$, so (since $0.8856\dots/4 \geq 1/5$)
 \[ |S| \geq \frac{1}{5\epsilon} . \]
Otherwise, and assuming $p < \infty$, we apply Stirling's approximation, $\left(\frac{k}{e}\right)^k\sqrt{2\pi k} \leq \Gamma\left(1+k\right) \leq e^{\frac{1}{12k}}\left(\frac{k}{e}\right)^k\sqrt{2\pi k}$, to both the numerator and denominator. We get
\begin{align*}
 |S| &\geq e^{\frac{p}{12}} \frac{\sqrt{2\pi\frac{\hat{n}-1}{p}}\left(\frac{\hat{n}-1}{pe}\right)^{\frac{\hat{n}-1}{p}}}{\sqrt{2\pi(\hat{n}-1)}\left(\frac{\hat{n}-1}{e}\right)^{\hat{n}-1}\left(4\Gamma\left(1+\frac{1}{p}\right) \epsilon\right)^{\hat{n}-1}}  \\
     &= e^{\frac{p}{12}} \frac{1}{\sqrt{p}} \left(\left(\frac{\hat{n}-1}{e}\right)^{\frac{1}{p}-1}\frac{1}{p^{\frac{1}{p}}}\frac{1}{4\Gamma\left(1+\frac{1}{p}\right) \epsilon}\right)^{\hat{n}-1}  \\
     &= e^{\frac{p}{12}} \frac{1}{\sqrt{p}} \left(\frac{1}{(\hat{n}-1)^{1/q} C_p \epsilon}\right)^{\hat{n}-1}
\end{align*}
where $C_p = 4\Gamma\left(1+\frac{1}{p}\right)p^{\frac{1}{p}}/e^{1/q}$, which (by maximizing over $p$) is at most $4$.
\end{proof}

The next step is to bound the entropy of the input samples.
\begin{lemma} \label{lemma:learning-entropy}
For any distribution $A$ on support size $\hat{n}$, the entropy of $\vec{X}$, the result of $m$ i.i.d. samples from $A$, is
  \[ H(\vec{X}) \leq \frac{\hat{n}-1}{2}\log\left(2\pi e \frac{m}{\hat{n}}\right) + \bigO{\frac{\hat{n}}{m}} .  \]
\end{lemma}
\begin{proof}
The samples consist of $\vec{X} = X_1,\dots,X_{\hat{n}}$ where $X_i$ is the number of samples drawn of coordinate $i$. Thus
\begin{align*}
 H(\vec{X})
  &= \sum_{i=1}^{\hat{n}} H(X_i \mid X_1,\dots,X_{i-1})  \\
  &= \sum_{i=1}^{\hat{n}-1} H(X_i \mid X_1,\dots,X_{i-1})  \\
  &\leq \sum_{i=1}^{\hat{n}-1} H(X_i)  \\
  &\leq \sum_{i=1}^{\hat{n}-1} \frac{1}{2}\log\left(2\pi e m A_i (1-A_i)\right) + \bigO{\frac{1}{m}}  \\
  &\leq \frac{\hat{n}-1}{2} \log\left(2\pi e \frac{m}{\hat{n}} \right) + \bigO{\frac{\hat{n}}{m}}.
\end{align*}
We used in the second line that the entropy of $X_{\hat{n}}$, given $X_1,\dots,X_{\hat{n}-1}$, is zero because it is completely determined (always equal to $m$ minus the sum of the previous $X_i$). Then, we plugged in the entropy of the Binomial distribution, as each $X_i \sim Binom(m,A_i)$. Then, we dropped the $(1-A_i)$ from each term, and used concavity to conclude that the uniform distribution $A_i = \frac{1}{\hat{n}}$ maximizes the bound. (We have glossed over a slight subtlety, that as stated the optimizer is uniform on coordinates $1,\dots,\hat{n}-1$. The full proof is to first note that \emph{any} one of the coordinates may be designated $X_n$ and dropped from the entropy sum, since it is determined by the others; in particular the largest may be. Maximizing the bound then results in the uniform distribution over all $\hat{n}$ coordinates, since any one with higher-than-average probability would be the one dropped.)
\end{proof}

To relate the entropy to the probability of success, we simply use Fano's Lemma, which is a basic inequality relating the probability of a correct guess of a parameter given data to the conditional entropy between the parameter and the data. It is proved in \emph{e.g.} Cover's text~\cite{cover2006elements}, and gives us the following lemma.
\begin{lemma} \label{lemma:fano}
The probability of $\delta$ of losing the distribution identification game is at least
 \[ \delta \geq 1 - \frac{H(\vec{X}) + 1}{\log|S|} . \]
where $\vec{X}$ is the set of input samples.
\end{lemma}
\begin{proof}
By Fano's Lemma recast into our terminology~\cite{cover2006elements},
 \[ \delta \geq \frac{H(A \mid \vec{X}) - 1}{\log|S|} . \]
If the distribution $A$ is selected uniformly from $S$, then
\begin{align*}
 H(A \mid \vec{X}) &= H(A,\vec{X}) - H(\vec{X})  \\
                   &\geq H(A) - H(\vec{X})  \\
                   &= \log|S| - H(\vec{X}) ,
\end{align*}
which proves the lemma.
\end{proof}

Now we can start combining our lemmas.
\begin{theorem} \label{theorem:learning-lower-with-S}
To win the distribution game with probability $1-\delta$ against a set $S$ with choice parameter $\hat{n}$ requires the following number of samples:
  \[  m = \bigOmega{ \hat{n} |S|^{\frac{2(1-\delta)}{\hat{n}-1}} } .  \]
\end{theorem}
\begin{proof}
Combining Lemmas \ref{lemma:fano} and \ref{lemma:learning-entropy},
\begin{align*}
 1 - \delta
  &<    \frac{H(\vec{X}) + 1}{\log|S|}  \\
  &\leq \frac{\frac{\hat{n}-1}{2}\log\left(2\pi e \frac{m}{\hat{n}}\right) + \bigO{\frac{\hat{n}}{m}}}{\log|S|} .
\end{align*}
Rearranging,
\begin{align*}
 \log\left(2\pi e \frac{m}{\hat{n}}\right) &\geq (1-\delta)\frac{2}{\hat{n}-1}\log|S| - \bigO{\frac{1}{m}}  \\
 \implies m &\geq \bigOmega{ \hat{n} |S|^{\frac{2(1-\delta)}{\hat{n}-1}} } .
\end{align*}
\end{proof}

We are now ready to prove the actual bounds.

\begin{theorem} \label{theorem:learning-lower-our-bounds}
To win the distribution identification game (and thus, by Lemma \ref{lemma:learning-id-game}, to learn in $\ell_p$ distance) with probability at least $1-\delta$, the number of samples required is at least
 \[ m = \begin{cases} \bigOmega{\frac{1}{\epsilon^{2(1-\delta)}}}  &  \text{unconditionally}  \\
                      \bigOmega{\frac{n}{(n^{1/q}\epsilon)^{2(1-\delta)}}}  &  \text{if $p < \infty$} \\
                      \bigOmega{\frac{1}{\epsilon^q}}  &  \text{if $p < \infty$, $n \geq \bigOmega{\frac{1}{\epsilon^q}}$}  . \end{cases} \]
\end{theorem}

\begin{proof}
By Lemma \ref{theorem:learning-lower-with-S}, we must have
  \[  m = \bigOmega{ \hat{n} |S|^{\frac{2(1-\delta)}{\hat{n}-1}} } .  \]
Now we make three possible choices of $\hat{n}$ and, for each, plug in the lower bound for $|S|$ from Corollary \ref{corollary:learning-size-S}. First, unconditionally, we may choose $\hat{n}=2$ and the bound $|S| \geq \frac{1}{5\epsilon}$, so
\begin{align*}
  m \geq \bigOmega{ \frac{1}{\epsilon^{2(1-\delta)}} } .
\end{align*}
Now suppose $p < \infty$. For both the second and third choices, we use the bound
 \[ |S| \geq \frac{e^{\frac{p}{12}}}{\sqrt{p}} \left(\frac{1}{4(\hat{n}-1)^{1/q}\epsilon}\right)^{\hat{n}-1} . \]
We get (hiding dependence on $p$ in the Omega):
\begin{align*}
 m &\geq \bigOmega{ \hat{n} \left(\frac{1}{\hat{n}^{1/q} \epsilon}\right)^{2(1-\delta)} } .
\end{align*}
To get the second case, we may always take $\hat{n} = n$. To get the third, if $n-1 \geq \frac{1}{\epsilon^q}$, then we may always take $\hat{n} = \frac{1}{\epsilon^q}$.
\end{proof}

We can improve the lower bound for the case $p \geq 2$ using the problem of distinguishing a biased coin from uniform.
\begin{theorem} \label{theorem:learning-lower-from-coin}
To learn in $\ell_p$ distance for for any $p$ (in particular $p \geq 2$) requires at least the following number of samples:
  \[ m = \frac{1}{16}\frac{\ln\left(1 + 2(1-2\delta)^2\right)}{\epsilon^2} . \]
\end{theorem}
\begin{proof}
If one can learn to within $\ell_p$ distance $\epsilon$, then one can test whether a distribution is $2\epsilon$-far from uniform in $\ell_{\infty}$ distance: Simply learn the distribution and output ``uniform'' if your estimate is within $\ell_{\infty}$ distance $\epsilon$ of $U_n$ (note that if we have learned to $\ell_p$ distance $\epsilon$, then we have also learned to $\ell_{\infty}$ distance $\epsilon$). This is correct by the triangle inequality. Therefore the lower bound for $\ell_{\infty}$ learning, Theorem \ref{theorem:test-unif-lower-infty-n}, applies with $n=2$ and $2\epsilon$ substituted for $\epsilon$.
\end{proof}

%% Omit %%%%%%%%%%%%%%%%%%%%%%%%%%%%%%%%%%%%%%%%%%%%%%%%%%

\fi  %%%%%%%%%%%%%%%%%%%%%%%%%%%%%%%% full version

\end{document}